\documentclass[twocolumn]{autart}    
\usepackage{graphicx}          
\usepackage{amsmath}
\usepackage{amssymb}
\usepackage{amsfonts}
\usepackage{mathptmx} 
\usepackage{times}
\usepackage{graphicx}
\usepackage{datetime}
\DeclareMathAlphabet{\bit}{OML}{cmm}{b}{it}

\def\proof{{\bf Proof. }}
\def\endproof{\hfill$\blacksquare$}

\def\vect{\mathrm{vec}}

\def\ad{\mathrm{ad}}           

\def\<{\leqslant}           
\def\>{\geqslant}           

\def\d{\partial}
\def\wh{\widehat}
\def\wt{\widetilde}

\def\Re{\mathrm{Re}}   
\def\Im{\mathrm{Im}}   

\def\cH{\mathcal{H}}   
\def\mA{\mathbb{A}}    
\def\mZ{\mathbb{Z}}    
\def\mR{\mathbb{R}}    
\def\mC{\mathbb{C}}    

\def\Tr{\mathrm{Tr}}       
\def\rT{\mathrm{T}}        
\def\bS{\mathbf{S}}

\def\bE{\mathbf{E}}    

\def\rprod{\mathop{\overrightarrow{\prod}}}

\def\bra{{\langle}}
\def\ket{{\rangle}}

\def\Bra{\left\langle}
\def\Ket{\right\rangle}

\def\re{\mathrm{e}}        
\def\rd{\mathrm{d}}        

\def\bJ{\mathbf{J}}
\def\sj{\mathsf{j}}

\def\br{\mathbf{r}}
\def\x{\times}
\def\ox{\otimes}
\def\op{\oplus}

\def\cZ{{\mathcal Z}}

\def\cX{\mathcal{X}}

\def\cK{\mathcal{K}}

\def\cC{\mathcal{C}}

\def\cI{\mathcal{I}}
\def\cP{\mathcal{P}}
\def\cQ{\mathcal{Q}}

\def\cA{\mathcal{A}}

\def\cE{\mathcal{E}}

\def\bL{\mathbf{L}}

\def\mH{\mathbb{H}}
\def\mS{\mathbb{S}}
\def\mT{\mathbb{T}}
\def\mZ{\mathbb{Z}}

\def\eps{\epsilon}
\def\Ups{\Upsilon}

\def\diag{\mathop{\mathrm{diag}}}    

\begin{document}

\begin{frontmatter}

\title{Direct Coupling Coherent Quantum Observers with Discounted Mean Square Performance
Criteria and Penalized Back-action\thanksref{footnoteinfo}} 

\thanks[footnoteinfo]{This work is supported 
by the Air Force Office of Scientific Research (AFOSR) under agreement number FA2386-16-1-4065 and the Australian Research Council under grant DP180101805. A brief version \cite{VP_2016} of this paper was presented at the IEEE 2016 Conference on Norbert Wiener in the 21st Century.
}

\author[ANU]{Igor G. Vladimirov}\ead{igor.g.vladimirov@gmail.com},\qquad    
\author[ANU]{Ian R. Petersen}\ead{i.r.petersen@gmail.com}               

\address[ANU]{Australian National University, Canberra, Australia}  


\begin{abstract}                          
This paper is concerned with quantum harmonic oscillators consisting of a quantum plant and a directly coupled coherent quantum observer. We employ discounted quadratic performance criteria in the form of exponentially weighted time averages of second-order moments of the system variables. Small-gain-theorem  bounds  are obtained for the back-action of the observer on the covariance dynamics of the plant in terms of the plant-observer coupling. A coherent quantum filtering (CQF) problem is formulated as the minimization of the discounted mean square of an estimation error, with which the dynamic variables of the observer approximate those of the plant. The cost functional also involves a quadratic penalty on the plant-observer coupling matrix in order to mitigate the back-action effect. For the discounted mean square optimal CQF problem with penalized back-action, we establish first-order necessary conditions of optimality in the form of algebraic matrix  equations. By using the Hamiltonian structure of the Heisenberg dynamics and Lie-algebraic techniques, this set of equations is represented in a more explicit form for equally dimensioned plant and observer. For a class of such observers with autonomous estimation error dynamics, we obtain a solution of the CQF problem and outline a homotopy method. The computation of the performance criteria  and the observer synthesis are illustrated by numerical examples.
\end{abstract}

\end{frontmatter}

\section{Introduction}

Noncommutative counterparts of classical control and filtering problems \cite{AM_1979,AM_1989,KS_1972} are a subject of active research
in quantum control which is concerned with  dynamical and stochastic systems governed by the laws of quantum mechanics and quantum probability \cite{H_2001,M_1995}. These developments (see, for example, \cite{JNP_2008,MJ_2012,NJP_2009,VP_2013a,VP_2013b}) are particularly focused on open quantum systems whose internal dynamics are affected by interaction with the environment \cite{BP_2006}. In such systems, the evolution of dynamic variables (as noncommutative operators on a Hilbert space) is often modelled using the Hudson-Parthasarathy calculus \cite{H_1991,HP_1984,P_1992} which provides a rigorous framework of quantum stochastic  differential equations (QSDEs) driven by quantum Wiener processes on symmetric Fock spaces.
In particular, linear QSDEs model open quantum harmonic oscillators (OQHOs) \cite{EB_2005} whose dynamic  variables (such as the position and momentum or annihilation and creation operators \cite{M_1998,S_1994}) satisfy canonical commutation relations (CCRs). This class of QSDEs is important for linear quantum control theory \cite{P_2010} and applications to quantum optics
\cite{GZ_2004,WM_1994_book} which provides one of platforms for quantum information  technologies \cite{NC_2000}.

One of the fundamental problems for quantum stochastic systems is the coherent quantum linear quadratic Gaussian (CQLQG)  control problem \cite{NJP_2009} which is a quantum mechanical counterpart of the classical LQG control problem. The latter is well-known in linear stochastic control theory due to the separation principle and its links with Kalman filtering and deterministic optimal control settings such as the linear quadratic regulator (LQR) problem \cite{AM_1989,KS_1972}. Coherent quantum feedback control \cite{L_2000,WM_1994_paper} employs the idea of control by interconnection, whereby quantum systems interact with each other directly or through optical fields in a measurement-free fashion, which can be described using the quantum feedback network formalism \cite{GJ_2009}.
In comparison with the traditional observation-actuation control paradigm, 
coherent quantum control avoids the ``lossy'' conversion  of operator-valued quantum variables into classical signals (which underlies the quantum measurement process), is potentially faster and can be implemented on micro and nano-scales using natural  quantum mechanical effects.

In coherent quantum filtering (CQF) problems \cite{MJ_2012,VP_2013b}, which are ``feedback-free'' versions of the CQLQG control problem, an observer is cascaded in a measurement-free fashion with a quantum plant so as to develop quantum correlations with the latter over the course of time. Both problems employ mean square performance criteria and involve physical realizability (PR) constraints \cite{JNP_2008,SP_2012} on the state-space matrices of the quantum controllers and filters. The PR constraints are a consequence of the specific Hamiltonian structure of quantum dynamics and complicate the  design of optimal coherent quantum controllers and filters. Variational approaches of \cite{VP_2011b}--\cite{VP_2013b}  reformulate the underlying problem as a constrained covariance control problem and employ an adaptation of ideas from dynamic programming, the Pontryagin minimum principle \cite{PBGM_1962,SW_1997} and nonlinear functional analysis. In particular, the Frechet differentiation of the LQG cost with respect to the state-space matrices of the controller or filter  subject to the PR constraints leads to necessary conditions of optimality in the form of nonlinear algebraic matrix equations. Although this approach is quite similar to \cite{BH_1989,SIG_1998} (with the quantum nature of the problem manifesting itself only through the PR constraints), the resulting equations appear to be much harder to solve than their classical predecessors.

Fully quantum variational techniques, using  perturbation analysis \cite{SVP_2014,V_2015a,V_2015b,V_2017}  
beyond the class of OQHOs and symplectic geometric tools \cite{SVP_2017}, suggest that the complicated sets of nonlinear equations for optimal quantum controllers  and filters may appear to be more amenable to solution if they are approached using Hamiltonian 
structures   similar to those in the underlying quantum dynamics. Such structures are particularly transparent in closed QHOs. Indeed, these models of linear quantum systems do not involve external bosonic fields and are technically simpler than the above mentioned OQHOs (although leave room for modelling the latter, for example,  through the Caldeira-Leggett infinite system limit
using a bath of harmonic oscillators \cite{CL_1981}).

We employ this class of models in the present paper and consider a mean square optimal CQF problem for a plant and a directly coupled observer which form a  closed QHO. Since this setting does not use
quantum Wiener processes, it simplifies the technical side of the treatment in comparison with \cite{MJ_2012,VP_2013b}.
The Hamiltonian  of the plant-observer QHO is a quadratic function of the dynamic variables satisfying the CCRs. 
When the energy matrix, which specifies the quadratic form of the Hamiltonian, is positive semi-definite, the system variables of the QHO  are either constant or exhibit oscillatory behaviour. This motivates the use of a cost functional (being minimized) in the form of a discounted mean square of an estimation error (with an exponentially decaying weight \cite{B_1965}) with which the observer variables approximate given linear combinations of the plant variables of interest. The performance criterion also involves a quadratic penalty on the plant-observer coupling 
in order to 
achieve a compromise between the conflicting requirements of minimizing the estimation error and reducing the back-action of the observer on the plant.  The CQF problem with penalized back-action can also be regarded as a quantum-mechanical counterpart to the classical LQR problem. The use of discounted averages of nonlinear moments of system variables and the presence of optimization makes this setting  different from the time-averaged approach of \cite{P_2014,P_2017} to CQF in directly coupled  QHOs (see \cite{PH_2015} for a quantum-optical implementation of that approach).

Since discounted moments of system variables for QHOs 
play an important role throughout the paper, we discuss the computation  of such moments in the state-space and frequency domains for completeness. Using the ideas of the small-gain theorem (see, for example, \cite{DJ_2006} and references therein) and linear matrix inequalities,   we establish upper bounds  for the back-action of the observer on the covariance dynamics of the plant in terms of the plant-observer coupling. This leads to a lower bound for the mean square of the estimation error in terms of its value for uncoupled plant and observer.
Similarly to the variational approach of \cite{VP_2013a,VP_2013b}, we develop first-order necessary conditions of optimality for the CQF problem being considered. These conditions are organized as a set of two algebraic Lyapunov equations (ALEs) for the controllability and observability Gramians which are coupled through another equation for the Hankelian (the product of the Gramians) of the plant-observer composite system. The Hamiltonian structure of the underlying Heisenberg dynamics allows Lie-algebraic techniques (in particular, the Jacobi identity \cite{D_2006}) to be employed in order to represent this set of equations in terms of the commutators of appropriately transformed  Gramians. This leads to a more tractable form of the optimality conditions for equally dimensioned plant and observer. 
We single out a class of such observers with
autonomous estimation error dynamics, for which the CQF problem is amenable to numerical solution through a homotopy method (similar to \cite{MB_1985}) over the penalty parameter. We also investigate the asymptotic behaviour of the resulting optimal observers in the weak-coupling limit, and illustrate the performance criteria computation and observer synthesis by numerical examples.

The paper is organised as follows. 
Section~\ref{sec:QHO} specifies the closed QHOs including its subclass  with positive semi-definite energy matrices. Section~\ref{sec:aver} describes the discounted averaging of moments for system operators in such QHOs in the time and frequency domains and illustrates their computation by a numerical example.
Section~\ref{sec:QODE} specifies the direct coupling of quantum plants and coherent quantum observers. Section~\ref{sec:back} discusses bounds for the observer back-action  on the covariance dynamics of the plant. Section~\ref{sec:CQF} formulates the discounted mean square optimal CQF problem with penalized back-action and discusses coupling-estimation inequalities. Section~\ref{sec:opt} establishes first-order necessary conditions of optimality for this problem. Section~\ref{sec:lie} represents the optimality conditions in a Lie-algebraic form. 
Section~\ref{sec:sub} provides a suboptimal solution of the CQF problem for a class of observers with autonomous estimation error dynamics and gives a numerical  example of observer synthesis.
Section~\ref{sec:conc} makes concluding  remarks.

\section{Quantum harmonic oscillators}\label{sec:QHO}

Consider a QHO \cite{M_1998}  with an even number $n$ of dynamic variables $X_1, \ldots, X_n$ which are time-varying self-adjoint operators on a  complex separable Hilbert space $\cH$ satisfying the CCRs
\begin{equation}
\label{XCCR}
    [X(t),X(t)^{\rT}]
    :=
    ([X_j(t),X_k(t)])_{1\< j,k\< n}
    =
    2i
    \Theta,
    \quad
    X
    :=
    {\scriptsize\begin{bmatrix}
        X_1(t)\\
        \vdots\\
        X_n(t)
    \end{bmatrix}}
\end{equation}
at any instant $t\> 0$ (the time arguments will often be omitted for brevity). It is assumed that the \emph{CCR matrix}  $\Theta \in \mA_n$ is nonsingular.  Here, $\mA_n$ denotes the subspace of real antisymmetric matrices  of order $n$. The entries $\theta_{jk}$ of $\Theta$ in (\ref{XCCR}) represent the scaling operators $\theta_{jk}\cI$, with $\cI$ the identity operator  on $\cH$. The transpose $(\cdot)^{\rT}$ acts on matrices of operators as if the latter were scalars, vectors are organized as columns unless indicated otherwise,  $[\phi,\psi]:= \varphi\psi-\psi\varphi$ is the commutator of operators, and $i:=\sqrt{-1}$ is the imaginary unit. 
The QHO has a quadratic Hamiltonian
\begin{equation}
\label{HR}
    H := \tfrac{1}{2}X^{\rT} R X,
\end{equation}
specified by an \emph{energy matrix} $R \in \mS_n$, with $\mS_n$ the subspace of real symmetric matrices of order $n$. Due to (\ref{XCCR}) and (\ref{HR}),  the Heisenberg dynamics of the QHO are governed by a linear ODE
\begin{equation}
\label{Xd}
    \dot{X} = i[H,X] = A X,
\end{equation}
where $A\in \mR^{n\x n}$  is a matrix of constant coefficients  given by
\begin{equation}
\label{A}
    A := 2\Theta R.
\end{equation}
The solution of the ODE (\ref{Xd}) is expressed using the standard matrix exponential as
\begin{equation}
\label{XA}
    X(t)
    =
    \sj_t(X_0)
    :=
    U(t)^{\dagger} X_0 U(t)
    =
    \re^{it \ad_{H_0}}(X_0)
    =
    \re^{tA} X_0,
\end{equation}
where
$\ad_{\alpha}:= [\alpha, \cdot]$, and  the subscript $(\cdot)_0$ indicates the initial values at time $t=0$.
The first three equalities in (\ref{XA}) apply to a general Hamiltonian $H_0$ (that is, not necessarily a quadratic function of $X_0$), and 
$  U(t):= \re^{-itH_0}$ 
is a time-varying unitary operator on $\cH$ (with the adjoint  $U(t)^{\dagger} = \re^{itH_0}$), which specifies the flow $\sj_t$ in (\ref{XA}) acting as a unitary similarity transformation on the system variables. The flow $\sj_t$ preserves the CCRs (\ref{XCCR}) which, in view of the relation
$
    [X(t),X(t)^{\rT}]
    =
    \re^{tA} [X_0, X_0^{\rT}] \re^{tA^{\rT}}
    =
    2i\re^{tA} \Theta \re^{tA^{\rT}}
    =
    2i\Theta
$,
are equivalent to the symplectic property $\re^{tA} \Theta \re^{tA^{\rT}} = \Theta$ of the matrix $\re^{tA}$ for any time $t\> 0$.  The infinitesimal form of this property is
$    A\Theta + \Theta A^{\rT} = 0
$.
This equality corresponds to the PR conditions for OQHOs \cite{JNP_2008,SP_2009} and its fulfillment is ensured by the Hamiltonian structure $A\in \Theta \mS_n$ of the matrix  $A$ in (\ref{A}). Similarly to classical linear systems, if the initial quantum state of the QHO is Gaussian \cite{CH_1971,KRP_2010}, it remains so over the course of time due to the deterministic linear dependence of $X(t)$ on $X_0$ in (\ref{XA}).
If the energy matrix in (\ref{HR}) is positive semi-definite, $R\succcurlyeq  0$ (and hence,  has a square root $\sqrt{R} \succcurlyeq 0$),  then $A=2\Theta \sqrt{R}\sqrt{R}$ is isospectral to the matrix $2\sqrt{R}\Theta \sqrt{R} \in \mA_n$ whose eigenvalues are purely imaginary \cite{HJ_2007}. In the case
 $R\succ 0$, this follows directly from the similarity transformation
\begin{equation}
\label{AR}
    A = R^{-1/2} (2\sqrt{R}\Theta \sqrt{R}) \sqrt{R}
\end{equation}
(see, for example, \cite{P_2014}), whereby
$A$ is diagonalized as
\begin{equation}
\label{AV}
    A = i V \Omega W,
    \qquad
    W := V^{-1},
    \qquad
    \Omega := \diag_{1\< k\< n}(\omega_k).
\end{equation}
Here, $W:= (w_{jk})_{1\< j,k\< n}\in \mC^{n\x n}$ is the inverse of a nonsingular matrix  $V:= (v_{jk})_{1\< j,k\< n}\in \mC^{n\x n}$ whose columns $V_1,\ldots, V_n\in \mC^n$ are the eigenvectors of $A$, and $\Omega:= \diag_{1\< k\< n}(\omega_k)\in \mR^{n \x n}$ is a diagonal matrix of frequencies of the QHO. These frequencies  (which should not be confused with the eigenvalues of the Hamiltonian $H$ as an operator on $\cH$ describing the energy levels of the QHO \cite{S_1994}) are nonzero and symmetric about the origin, and, without loss of generality, are assumed to be arranged so that
\begin{equation}
\label{sym}
    \omega_k = -\omega_{k + \tfrac{n}{2}} >0,
    \qquad
    k=1, \ldots, \tfrac{n}{2}.
\end{equation}
Note that  $\sqrt{R} V$ is a unitary matrix  whose columns are the eigenvectors of the matrix $i\sqrt{R}\Theta \sqrt{R} \in \mH_n$ in view of (\ref{AR}); see also the proof of Williamson's symplectic diagonalization  theorem \cite{W_1936,W_1937} in \cite[pp. 244--245]{D_2006}. Here, $\mH_n$ denotes the subspace of complex Hermitian matrices of order $n$. Substitution of (\ref{AV}) into (\ref{XA}) leads to
\begin{equation}
\label{XAV}
    X(t) = V \re^{it\Omega} W X_0.
\end{equation}
Due to the presence of the matrix $\re^{it\Omega} = \diag_{1\< k\< n}(\re^{i\omega_k t})$ in (\ref{XAV}), the  dynamic variables of the QHO are linear combinations of their initial values whose coefficients are trigonometric polynomials of time:
\begin{equation}
\label{Xt0}
    X_j(t) = \sum_{k,\ell=1}^n c_{jk\ell}\re^{i\omega_k t} X_{\ell}(0),
    \qquad
    j=1, \ldots, n,
\end{equation}
where $c_{jk\ell}$ are complex parameters which are assembled into rank-one matrices
\begin{equation}
\label{C}
     C_k := (c_{jk\ell})_{1\< j,\ell \< n} = V_kW_k,
    \qquad
    c_{jk\ell}:=  v_{jk} w_{k\ell},
\end{equation}
with $W_k$  denoting the $k$th row of $W$. The matrices $C_1, \ldots, C_n$ form a resolution of the identity: $\sum_{k=1}^n C_k  = VW=I_n$. Also,
\begin{equation}
\label{CC}
     \overline{C_k} = C_{k + \tfrac{n}{2}},
     \qquad
    k = 1, \ldots, \tfrac{n}{2},
\end{equation}
in accordance with (\ref{sym}), whereby (\ref{Xt0}) can be represented in vector-matrix form as
\begin{equation}
\label{XCX}
    X(t)
    =
    \sum_{k=1}^{n/2}
    \big(
    \re^{i\omega_k t} C_k
    +
    \re^{-i\omega_k t} \overline{C_k}
    \big)
    X_0
    =
    2\sum_{k=1}^{n/2}\Re(\re^{i\omega_k t} C_k) X_0,
\end{equation}
where $\overline{(\cdot)}$ is the complex conjugate.
Therefore, for any positive integer $d$ and any $d$-index $j:= (j_1, \ldots, j_d)\in \{1, \ldots, n\}^d$, the following degree $d$ monomial of the system variables is also a trigonometric polynomial of time $t$:
\begin{equation}
\label{Xit}
    \Xi_j(t)
    :=
    \rprod_{s=1}^{d}
    X_{j_s}(t)
    =
    \sum_{k,\ell \in \{1, \ldots, n\}^d}\,
    \prod_{s=1}^{d}
    c_{j_sk_s\ell_s}
    \re^{i\omega_{k_s}t}\,
    \Xi_{\ell}(0).
\end{equation}
Here, $\rprod$ denotes the ``rightwards'' ordered product of operators (the order of multiplication is essential for non-commutative quantum variables),
and the sum is taken over $d$-indices $k:= (k_1, \ldots, k_d), \ell:= (\ell_1, \ldots, \ell_d)\in \{1, \ldots, n\}^d$. Note that (\ref{Xt0}) is a particular case of (\ref{Xit}) with $d=1$. The relations (\ref{XAV})--(\ref{Xit}) remain valid in the case of $R\succcurlyeq 0$, except that (\ref{sym}) is relaxed to the frequencies $\omega_1, \ldots, \omega_{n/2}$ being nonnegative.

\section{Discounted moments of system operators}\label{sec:aver}

For any $\tau>0$, we define a linear functional $\bE_{\tau}$ which maps a system operator $\sigma$ of the QHO to the weighted time average
\begin{equation}
\label{bEtau}
    \bE_{\tau}\sigma
    :=
    \tfrac{1}{\tau}
    \int_0^{+\infty}
    \re^{-t/\tau}
    \bE \sigma(t)
    \rd t.
\end{equation}
Here, $\bE\sigma:= \Tr(\rho \sigma)$  denotes the quantum expectation over the underlying quantum state $\rho$ (which is a positive semi-definite self-adjoint operator on $\cH$ with unit trace). The weighting function  $\tfrac{1}{\tau} \re^{-t/\tau}$  in (\ref{bEtau}) is the density of an exponential probability distribution with mean value $\tau$. Therefore, $\tau$ plays the role of an effective horizon for averaging $\bE\sigma$ over time.  This time average (where the relative importance of the quantity of interest decays exponentially)  has the structure of a discounted cost functional in dynamic programming problems \cite{B_1965}. In particular, if $\bE\sigma(t)$, as a function of time $t\>0$,  is right-continuous at $t=0$, then  $\lim_{\tau\to 0+} \bE_{\tau} \sigma = \bE \sigma_0$. At the other extreme, the \emph{infinite-horizon average} of $\sigma$ is defined by
\begin{equation}
\label{bEinf}
    \bE_{\infty} \sigma
    :=
    \lim_{\tau\to +\infty} \bE_{\tau}\sigma
    =
    \lim_{\tau\to +\infty}
    \Big(
        \tfrac{1}{\tau}
        \int_0^\tau
        \bE \sigma(t)
        \rd t
    \Big),
\end{equation}
provided these limits exist. The second of these equalities, whose right-hand side is the Cesaro mean of $\bE \sigma$, follows from the integral version of the Hardy-Littlewood Tauberian theorem \cite{F_1971}. In particular, (\ref{bEinf}) implies that $|\bE_{\infty}\sigma|\< \limsup_{t\to +\infty}|\bE \sigma(t)|$.

In the case when the QHO has a positive semi-definite energy matrix, the coefficients in (\ref{XCX}) and (\ref{Xit}) are either constant or oscillatory, which makes the time averages (\ref{bEtau}) and (\ref{bEinf}) well-defined for nonlinear functions of the system variables and their moments for any $\tau>0$. A similar property underlies applications of harmonic analysis to the heterodyne detection of signals. 
To this end, we will use the characteristic function $\chi_{\tau}: \mR\to \mC$ of the exponential distribution and its pointwise convergence:
\begin{align}
\nonumber
    \chi_{\tau}(u)
    & :=  \tfrac{1}{\tau}
    \int_0^{+\infty}
    \re^{-t/\tau}
    \re^{iut}
    \rd t
     =
    \tfrac{1}{1-iu \tau}\\
\label{chi}
     & \to
    \delta_{u0}
     =
    {\scriptsize\left\{
        {\begin{matrix}
            1 & {\rm if}\  u =0\\
            0 & {\rm if}\  u \ne 0
        \end{matrix}}
    \right.},
    \qquad
    {\rm as}\
    \tau\to +\infty,
\end{align}
where $\delta_{pq}$ is the Kronecker delta.
A combination of (\ref{Xit}) with (\ref{chi}) implies that if the initial system variables of the QHO have finite mixed moments $\bE \Xi_{\ell}(0)$ of order $d$ for all $\ell \in \{1,\ldots, n\}^d$, then such moments have the following time-averaged values (\ref{bEtau}):
\begin{align}
\nonumber
    \bE_{\tau} \Xi_j
    :=&
    \tfrac{1}{\tau}
    \int_0^{+\infty}
    \re^{-t/\tau}
    \bE \Xi_j(t)
    \rd t\\
\label{EXi}
    = &
    \sum_{k\in \{1, \ldots, n\}^d}
    \chi_{\tau}\Big(\sum_{s=1}^d \omega_{k_s}\Big)\!\!
    \sum_{\ell \in \{1, \ldots, n\}^d}\,
    \prod_{s=1}^{d}
    c_{j_sk_s\ell_s}
    \bE \Xi_{\ell}(0)
\end{align}
for any $j\in \{1, \ldots, n\}^d$. Hence, the corresponding infinite-horizon average (\ref{bEinf}) takes the form
\begin{equation}
\label{EXiinf}
    \bE_{\infty} \Xi_j
    =
    \sum_{k\in \cK_d}\,
    \sum_{\ell \in \{1, \ldots, n\}^d}\,
    \prod_{s=1}^{d}
    c_{j_sk_s\ell_s}
    \bE \Xi_{\ell}(0),
\end{equation}
where
$
    \cK_d
    :=
    \big\{
        (k_1, \ldots, k_d)\in \{1, \ldots, n\}^d:\
        \sum_{s=1}^{d}\omega_{k_s} = 0
    \big\}
$
is a subset of $d$-indices associated with the frequencies $\omega_1, \ldots, \omega_n$ of the QHO from (\ref{AV}).
For every even $d$, the set $\cK_d$ is nonempty due to the central symmetry of the frequencies.
 If the QHO is in a Gaussian quantum state, mentioned in Section~\ref{sec:QHO}, then the higher-order moments on the right-hand sides of (\ref{EXi}) and (\ref{EXiinf}) can be expressed in terms of the first two moments (with $d\<2$) by using the Isserlis-Wick theorem \cite{I_1918,J_1997}. However, the Gaussian assumption will not be employed in what follows.

The linear functional $\bE_{\tau}$ in (\ref{EXi}) and its limit $\bE_{\infty}$ in (\ref{EXiinf}) are extendable to polynomials and more general functions $\sigma:= f(X)$ of the system variables, provided $X_0$ satisfies appropriate integrability conditions. Such an extension of $\bE_{\infty}$, which involves the Cesaro mean, is similar to the argument used in the context of Besicovitch spaces of almost periodic functions \cite{B_1954}.  
If the system is in an invariant state $\rho$ (which, therefore, satisfies $[H_0, \rho] = 0$), then the quantum expectation $\bE \sigma = \Tr( \rho \re^{it\ad_{H_0}}(\sigma_0)) = \Tr(\re^{-it\ad_{H_0}}(\rho)\sigma_0) = \Tr(\rho\sigma_0) $ is time-independent for any system operator $\sigma_0$ evolved by the flow (\ref{XA}). In this case, the time averaging in (\ref{bEtau}) becomes redundant. 
However, the subsequent discussion is concerned with general (not necessarily invariant) quantum states $\rho$.

The following theorem is, in essence,  an adaptation of classical results on the averaging of quasi-periodic motions in Hamiltonian systems; see, for example, \cite[pp. 285--289]{A_1989}.

\begin{thm}
\label{th:torus}
Suppose the energy matrix in (\ref{HR}) satisfies $R\succ  0$. Also, let the frequencies $\omega_1, \ldots, \omega_{n/2}$ of the QHO, arranged according to (\ref{sym}), be incommensurable in the sense of rational independence (that is, their linear combination $\sum_{k=1}^{n/2}\lambda_k \omega_k$ with integer coefficients $\lambda_1, \ldots, \lambda_{n/2} \in \mZ$ vanishes if and only if $\lambda_1= \ldots= \lambda_{n/2}=0$). Furthermore, suppose a function $f:\mR^n \to \mR$ is extended to quantum variables so that
\begin{equation}
\label{g}
    g(\varphi)
    :=
    \bE f
    \Big(
        2 \sum_{k=1}^{n/2}\Re(\re^{i\varphi_k} C_k) X_0
    \Big)
\end{equation}
depends continuously on the phases $\varphi:= (\varphi_k)_{1\< k\< n/2} \in \mT^{n/2}$, where $\mT$ is the one-dimensional torus implemented as the interval $[0,2\pi)$, and the matrices $C_k$ are given by (\ref{C}). Then the system operator $f(X)$ has the following infinite-horizon average value (\ref{bEinf}):
\begin{equation}
\label{EfX}
    \bE_{\infty} f(X) = (2\pi)^{-n/2}\int_{\mT^{n/2}}g(\varphi) \rd \varphi.
\end{equation}
\end{thm}
\begin{proof}
From (\ref{XCX}) and (\ref{g}), it follows  that
$
    \bE f(X(t)) = g(t\mho)
$
for any $t\> 0$, where $\mho:= (\omega_k)_{1\< k \< n/2}$, and the entries of the vector $t \mho \in \mR^{n/2}$ are considered modulo $2\pi$. Since the function $g$, which is $2\pi$-periodic  in each of its variables, is assumed to be continuous  (and hence, $g$ is bounded due to the compactness of the torus), then  (\ref{EfX}) is established as
$
    \bE_{\infty} f(X)
    =
    \lim_{\tau\to +\infty}
    \Big(
        \tfrac{1}{\tau}
        \int_0^\tau
        g(t \mho)
        \rd t
    \Big)
    =
(2\pi)^{-n/2}\int_{\mT^{n/2}}g(\varphi) \rd \varphi
$.
The last equality is obtained by applying the Weyl equidistribution   criterion \cite{B_1968} to  the map $t\mapsto t\mho$ considered modulo $2\pi$, whereby its sample distribution
$D_{\tau}(S):= \tfrac{1}{\tau}\mu_1 \{0\< t\< \tau:\  t\mho \in S + \mZ^{n/2}\}$ for $S\subset \mT^{n/2}$
converges weakly  to the uniform probability measure
$(2\pi)^{-n/2}\mu_{n/2}(S)$
on the torus $\mT^{n/2}$, provided the frequencies are incommensurable.
More precisely, $\lim_{\tau\to +\infty} D_{\tau}(S) = (2\pi)^{-n/2}\mu_{n/2}(S)$ for any Borel set $S\subset \mT^{n/2}$ whose boundary $\d S$ satisfies $\mu_{n/2}(\d S) = 0$, where $\mu_r$ denotes the $r$-dimensional Lebesgue measure.
\end{proof}

The vectors $\mho$ of commensurable frequencies are contained in a denumerable union $\bigcup_{\lambda \in \mZ^{n/2}\setminus \{0\}} \lambda^{\bot}$ of the hyperplanes $\lambda^{\bot}:= \{\mho \in \mR^{n/2}:\, \lambda^{\rT} \mho = 0\}$ which has zero $n/2$-dimensional Lebesgue measure; see also \cite[p. 290]{A_1989}. Therefore, Theorem~\ref{th:torus}  applies to the infinite-horizon averaging of nonlinear functions of system variables in
QHOs with generic spectra.

Of particular use for our  purposes is the following lemma on state-space computation of the discounted time average (\ref{bEtau}) for second moments of the system variables, which is concerned with finite values of $\tau$ and does not employ
the imaginarity of the spectrum of $A$ and the frequency incommensurability condition of Theorem~\ref{th:torus}.
To this end, we note that $\bE(XX^{\rT})\in \mH_n^+$  at every moment of time due to
the 
generalized Heisenberg uncertainty principle \cite{H_2001}, where $\mH_n^+$ denotes the set of complex positive semi-definite Hermitian matrices of order $n$. Furthermore, $\Im \bE(XX^{\rT}) = \Theta$ remains unchanged in view of the preservation of the CCRs (\ref{XCCR}) mentioned above. Also, with any  Hurwitz matrix $\alpha$, we associate a linear operator $\bL(\alpha, \cdot)$
which maps an appropriately dimensioned matrix $\beta$ to a unique solution $\gamma=\bL(\alpha,\beta)$  of the
ALE $\alpha \gamma +\gamma\alpha^{\rT}+\beta= 0$:
\begin{equation}
\label{ILO}
  \bL(\alpha,\beta):= \int_0^{+\infty} \re^{t\alpha}\beta \re^{t\alpha^{\rT}}\rd t.
\end{equation}
The monotonicity of the operator $\bL(\alpha, \cdot)$ (with respect to the partial ordering induced by positive semi-definiteness) implies that
\begin{align}
\nonumber
    \bL(\alpha, \bL(\alpha, \beta))
     & =
    \bL(\alpha, \sqrt{\beta}\beta^{-1/2}\bL(\alpha, \beta)\beta^{-1/2}\sqrt{\beta})\\
\label{LL}
     & \preccurlyeq
    \br(\bL(\alpha, \beta)\beta^{-1})\, \bL(\alpha, \beta)
\end{align}
for any $\beta\succ 0$, where $\br(\cdot)$ denotes the spectral radius of a matrix, and use is made of the similarity transformation $\beta^{-1/2}N\beta^{-1/2}\mapsto 
N\beta^{-1}$.

\begin{lem}
\label{lem:PALE}
Let the initial dynamic variables of the QHO have finite second moments (that is, $\bE(X_0^{\rT}X_0)<+\infty$) whose real parts form the matrix
\begin{equation}
\label{Sigma}
  \Sigma:= \Re \bE(X_0X_0^{\rT}).
\end{equation}
Also, suppose the effective time horizon $\tau>0$ is bounded above as
\begin{equation}
\label{taumax}
    \tau < \tfrac{1}{2\max(0,\, \ln \br(\re^A))}.
\end{equation}
Then the matrix of the real parts of the discounted second moments of the dynamic variables can be computed as
\begin{equation}
\label{P}
    P:= \Re \bE_{\tau}(XX^{\rT}) = \tfrac{1}{\tau}\bL(A_{\tau},\Sigma)
\end{equation}
through the operator (\ref{ILO}). That is, $P$ is a unique solution of the ALE
\begin{equation}
\label{PALE}
    A_{\tau}
    P
    +
    P
    A_{\tau}^{\rT} + \tfrac{1}{\tau}\Sigma= 0,
\end{equation}
with the Hurwitz matrix
\begin{equation}
\label{Atau}
    A_{\tau}:= A - \tfrac{1}{2\tau}I_n.
\end{equation}
\end{lem}
\begin{proof}
By combining (\ref{XA}) with (\ref{Sigma}), it follows that $\Re \bE(X(t)X(t)^{\rT}) = \re^{tA} \Sigma \re^{tA^{\rT}}$ for any $t\>0$. Hence, in application to the matrix $P$ in (\ref{P}),     the time average (\ref{bEtau}) can be computed  as
$
    P
     =
    \tfrac{1}{\tau}
    \int_0^{+\infty}
    \re^{-t/\tau}
    \Re \bE(X(t)X(t)^{\rT})
    \rd t
      =
    \tfrac{1}{\tau}
    \int_0^{+\infty}
    \re^{-t/\tau}
    \re^{tA} \Sigma \re^{tA^{\rT}}
    \rd t
      =
    \tfrac{1}{\tau}
    \int_0^{+\infty}
    \re^{tA_{\tau}} \Sigma \re^{tA_{\tau}^{\rT}}
    \rd t
    =\tfrac{1}{\tau}\bL(A_{\tau},\Sigma)
$,
thus establishing the representation (\ref{P}).
Here, the matrix $A_{\tau}$, given by (\ref{Atau}), is Hurwitz due to the condition (\ref{taumax}).
\end{proof}

In view of (\ref{PALE}), the matrix $P$ is the controllability Gramian \cite{KS_1972} of the pair $(A_{\tau}, \sqrt{\tau^{-1}\Sigma})$.
In contrast to similar ALEs for steady-state covariance matrices in dissipative OQHOs \cite{EB_2005} (where the corresponding matrix $A$ itself is Hurwitz), the term   $\tfrac{1}{\tau}\Sigma$ in (\ref{PALE}) comes from the initial condition (\ref{Sigma}) instead of the Ito matrix of the quantum Wiener process \cite{H_2001,HP_1984,P_1992}. 
Since $A$ is a Hamiltonian matrix (and hence, its spectrum is symmetric about the imaginary axis), the condition (\ref{taumax}) is equivalent to the eigenvalues  of $A$ being contained in the strip $\big\{z \in \mC:\ |\Re z|< \tfrac{1}{2\tau}\big\}$.
For any $\tau>0$ satisfying (\ref{taumax}), a frequency-domain representation of the matrix $P$ in (\ref{P}) is
\begin{align}
\nonumber
    P
     & =
    \tfrac{1}{2\pi\tau}
    \Re
    \int_{-\infty}^{+\infty}
    F\Big(\tfrac{1}{2\tau} + i\omega\Big)
    \Gamma
    F
    \Big(\tfrac{1}{2\tau} + i\omega\Big)^*
    \rd \omega\\
\label{PLap}
    & =
    \tfrac{1}{2\pi\tau}
    \Im
    \int_{\Re s=\tfrac{1}{2\tau}}
    F(s)
    \Gamma
    F(s)^*
    \rd s,
\end{align}
with $(\cdot)^*:= (\overline{(\cdot)})^{\rT}$ the complex conjugate transpose.
Here,
\begin{equation}
\label{Gamma}
    \Gamma:= \bE(X_0 X_0^{\rT}) = \Sigma +i\Theta
\end{equation}
is the matrix of second moments of the initial system variables, and
\begin{equation}
\label{F}
    F(s):= (sI_n-A)^{-1}
\end{equation}
is the transfer function (with the complex variable $s$ satisfying     $\Re s > \ln\br(\re^A)$) which relates
the Laplace transform
\begin{equation}
\label{wtX}
    \wt{X}(s)
    :=
    \int_0^{+\infty}
    \re^{-st} X(t)\rd t
\end{equation}
of the quantum process $X$ from (\ref{XA}) to its initial value $X_0$ as
$
    \wt{X}(s)
    =
    \int_0^{+\infty}
    \re^{-t(sI_n-A)}\rd tX_0  = F(s) X_0
$.
The representation (\ref{PLap}) is obtained by applying an operator version of the Plancherel theorem to the inverse Fourier transform
$
    \re^{-\tfrac{t}{2\tau}} X(t)
    =
    \tfrac{1}{2\pi}
    \int_{-\infty}^{+\infty}
    \re^{i\omega t}
    \wt{X}\big(\tfrac{1}{2\tau} + i\omega\big)
    \rd \omega$ for
$
    t\> 0
$
under the condition (\ref{taumax}).

In the case $R\succcurlyeq  0$ (when the matrix $A$ in (\ref{A}) has a purely imaginary spectrum and (\ref{taumax}) holds for any arbitrarily large $\tau$), the formal limit of the ALE (\ref{PALE}), as $\tau\to +\infty$,  is $AP+PA^{\rT} = 0$, which does not have a unique solution. This non-uniqueness
makes the ALE approach of Lemma~\ref{lem:PALE} inapplicable to computing $\bE_{\infty}(XX^{\rT})$. We will therefore provide an alternative   calculation of the discounted second moments for completeness.

\begin{lem}
\label{lem:EXX}
Suppose the energy matrix $R$ of the QHO in (\ref{HR}) satisfies $R\succ 0$, and the initial dynamic variables have finite second moments assembled into the matrix $\Gamma$ in (\ref{Gamma}).
Then for any $\tau>0$,
\begin{align}
\nonumber
    \bE_{\tau} (XX^{\rT})
    = & V (\Phi_{\tau} \odot (W\Gamma W^*)) V^*\\
\nonumber
    = &
    \sum_{j,k=1}^{n/2}
    {\scriptsize\begin{bmatrix}
        V_j & \overline{V_j}
    \end{bmatrix}}\left(
    {\scriptsize\begin{bmatrix}
        \chi_{\tau}(\omega_j-\omega_k)  & \chi_{\tau}(\omega_j+\omega_k)\\
        \chi_{\tau}(-\omega_j-\omega_k) & \chi_{\tau}(\omega_k-\omega_j)
    \end{bmatrix}}\right.\\
\label{EXX}
     & \qquad \odot
     \left.
    \Big(
    {\scriptsize\begin{bmatrix}
        W_j \\
        \overline{W_j}
    \end{bmatrix}}
    \Gamma
    {\scriptsize\begin{bmatrix}
        W_k^* & W_k^{\rT}
    \end{bmatrix}}
    \Big)
    \right)
    {\scriptsize\begin{bmatrix}
        V_k^* \\
        V_k^{\rT}
    \end{bmatrix}}.
\end{align}
Here, $V$ is the matrix from (\ref{AV}), use is made of  an auxiliary matrix
\begin{equation}
\label{Phitau}
    \Phi_{\tau}
    :=
    (
        \chi_{\tau}(\omega_j-\omega_k)
    )_{1\< j,k\< n}
\end{equation}
associated with the frequencies of the QHO through the function $\chi_{\tau}$ from (\ref{chi}),  $\odot$ denotes the Hadamard product of matrices \cite{HJ_2007}, and $C_k$ are the matrices from (\ref{C}) satisfying (\ref{CC}) under the convention (\ref{sym}). Furthermore, the infinite-horizon time averages of the second moments are computed as
\begin{align}
\nonumber
    \bE_{\infty} (XX^{\rT})
    &= V (\Phi_{\infty} \odot (W\Gamma W^*)) V^*\\
\nonumber
    & = \sum_{j,k=1}^{n/2}
    \delta_{\omega_j\omega_k}
    \big(
    C_j \Gamma C_k^*
    +
    \overline{C_j} \Gamma C_k^{\rT}
    \big)\\
\label{EXXinf}
    & =
    \sum_{j,k=1}^{n/2}
    \delta_{\omega_j\omega_k}
    {\scriptsize\begin{bmatrix}
        V_j & \overline{V_j}
    \end{bmatrix}}
    {\scriptsize\begin{bmatrix}
        W_j \Gamma W_k^* & 0  \\
        0   & \overline{W_j} \Gamma W_k^{\rT}
    \end{bmatrix}}
    {\scriptsize\begin{bmatrix}
        V_k^* \\
        V_k^{\rT}
    \end{bmatrix}},
\end{align}
where use is made of a binary matrix
\begin{equation}
\label{Phiinf}
    \Phi_{\infty}
    :=
    (
        \delta_{\omega_j\omega_k}
    )_{1\< j,k\< n}.
\end{equation}
\end{lem}
\begin{proof}
Although (\ref{EXX}) can be obtained from the relation (\ref{EXi}) with $d=2$, we will provide a direct calculation. In view of self-adjointness of the system variables, (\ref{XAV}) and (\ref{XCX}) imply that
\begin{align}
\nonumber
    X&(t) X(t)^{\rT}
     =
    X(t)X(t)^{\dagger}\\
\nonumber
    & =
    V \re^{it\Omega} W X_0X_0^{\rT} W^* \re^{-it\Omega} V^*\\
\nonumber
     & =
    V \big(\Psi(t)\odot (W X_0X_0^{\rT} W^*)\big)V^*\\
\label{XXV}
    & =
    \sum_{j,k=1}^{n/2}
    \big(
    \re^{i\omega_j t} C_j
    +
    \re^{-i\omega_j t} \overline{C_j}
    \big)    X_0X_0^{\rT}
    \big(
    \re^{-i\omega_k t} C_k^*
    +
    \re^{i\omega_k t} C_k^{\rT}
    \big),
\end{align}
with $(\cdot)^{\dagger}:= ((\cdot)^{\#})^{\rT}$ the transpose of the entry-wise operator adjoint $(\cdot)^{\#}$.
Here, use is also made of the diagonal structure of the matrix $\Omega$ in (\ref{AV}) together  with a complex Hermitian rank-one matrix
\begin{equation}
\label{Ft}
    \Psi(t):=   (\re^{i(\omega_j - \omega_k)t})_{1\< j,k \< n}
\end{equation}
which encodes the time dependence of $XX^{\rT}$. The representation (\ref{XXV}) allows the time averaging to be decoupled from the quantum expectation as
\begin{align}
\nonumber
    \bE_{\tau}(X&X^{\rT})
    :=
    \tfrac{1}{\tau}
    \int_0^{+\infty}
    \re^{-t/\tau}
    \bE (X(t)X(t)^{\rT})
    \rd t\\
\nonumber
    = &
    V
    \Big(
        \tfrac{1}{\tau}
        \int_0^{+\infty}
        \re^{-t/\tau}
        \Psi(t)
        \rd t
        \odot (W \Gamma W^*)
    \Big)V^*\\
\nonumber
    =&
    \sum_{j,k=1}^{n/2}
    \Big(
        \chi_{\tau}(\omega_j-\omega_k) C_j \Gamma C_k^*
        +
        \chi_{\tau}(\omega_j+\omega_k) C_j \Gamma C_k^{\rT}\\
\label{EXXtau}
         & +
        \chi_{\tau}(-\omega_j-\omega_k) \overline{C_j} \Gamma C_k^*
        +
        \chi_{\tau}(\omega_k-\omega_j) \overline{C_j} \Gamma C_k^{\rT}
    \Big),
\end{align}
which leads to (\ref{EXX}) and (\ref{Phitau}) in view of (\ref{C}). Here,  $\Gamma$ is the matrix given by (\ref{Gamma}),  and the relation
$
        \tfrac{1}{\tau}
        \int_0^\tau
        \re^{-t/\tau}
        \Psi(t)
        \rd t
        = \Phi_{\tau}$
is obtained by applying (\ref{chi}) entrywise to the matrix $\Psi$ in (\ref{Ft}). Now, the convergence in (\ref{chi}) implies that the matrices (\ref{Phitau}) and (\ref{Phiinf}) are related by $\lim_{\tau\to +\infty} \Phi_{\tau} = \Phi_{\infty}$, and $\lim_{\tau\to +\infty}\chi_{\tau}(\pm(\omega_j+\omega_k))=0$ since $\omega_j + \omega_k>0$ for all $j,k=1, \ldots, \tfrac{n}{2}$  in view of (\ref{sym}). This leads to (\ref{EXXinf}) in view of (\ref{EXXtau}).
\end{proof}

The proof of Lemma~\ref{lem:EXX} shows that $\bE_{\tau}(XX^{\rT})$ is close to $\bE_{\infty}(XX^{\rT})$ if the effective time horizon $\tau$ is large in comparison with
\begin{equation}
\label{tau*}
    \tau_*:=
    \tfrac{1}{\min\big(\big\{|\omega_j\pm\omega_k |:\ 1\< j, k\< \tfrac{n}{2}\big\}\setminus \{0\}\big)}.
\end{equation}
If the frequencies $\omega_1, \ldots, \omega_{n/2}$ are pairwise different (which is a weaker condition than their incommensurability used in Theorem~\ref{th:torus}), the matrix $\Phi_{\infty}$ in (\ref{Phiinf}) becomes the identity matrix and (\ref{EXXinf}) reduces to
\begin{align}
\nonumber
    \bE_{\infty} (XX^{\rT})
    & =
    \sum_{k=1}^{n/2}
    \big(
    C_k \Gamma C_k^*
    +
    \overline{C_k} \Gamma C_k^{\rT}
    \big)\\
\label{EXX_red}
     & =
    \sum_{k=1}^{n/2}
    {\scriptsize\begin{bmatrix}
        V_k & \overline{V_k}
    \end{bmatrix}}
    {\scriptsize\begin{bmatrix}
        W_k \Gamma W_k^* & 0\\
        0 & \overline{W_k} \Gamma W_k^{\rT}
    \end{bmatrix}}
    {\scriptsize\begin{bmatrix}
        V_k^* \\
        V_k^{\rT}
    \end{bmatrix}}.
\end{align}
Such energy matrices $R\succ 0$ form an open subset  of $\mS_n$.  The corresponding infinite-horizon average of a quadratic form of $X$ is
$
    \bE_{\infty}(X^{\rT} \Pi X)
    =
    \sum_{k=1}^{n/2}
    \big(
        V_k^* \Pi V_k
        W_k \Gamma W_k^*
        +
        V_k^{\rT} \Pi \overline{V_k}
        \overline{W_k} \Gamma W_k^{\rT}
    \big)
$
for any $\Pi \in \mS_n$.
Lemmas~\ref{lem:PALE} and \ref{lem:EXX} can be used for computing quadratic cost functionals for QHOs, such  as the performance criterion  in the mean square optimal CQF problem of Section~\ref{sec:CQF}. Furthermore, Lemma~\ref{lem:EXX} can be easily extended to the more general case $R\succcurlyeq 0$ with nonstrict inequalities $\>$ in (\ref{sym}).

\noindent{\bf Example 1.}
Consider a two-mode QHO of dimension $n=4$, whose CCR matrix  is $\Theta := \tfrac{1}{2}I_2\ox \bJ$, where $\ox$ is the Kronecker product, and
\begin{equation}
\label{bJ}
\bJ:={\scriptsize\begin{bmatrix}
        0& 1\\
        -1 & 0
    \end{bmatrix}}
\end{equation}
spans the space $\mA_2$. This corresponds to the system variables consisting of two pairs of conjugate position $q_k$ and momentum $-i\d_{q_k}$ operators (with an appropriately normalised  Planck constant \cite{S_1994}),  
$k=1,2$. Suppose the QHO has the energy matrix
    $$
        R :={\scriptsize\begin{bmatrix}
    3.4048  &  3.0478 &  -2.2402 &  -1.4028\\
    3.0478  &  4.1266 &  -2.0050 &  -2.4614\\
   -2.2402  & -2.0050 &   2.0076 &   0.8484\\
   -1.4028  & -2.4614 &   0.8484 &   4.7504
\end{bmatrix}}\succ 0,
$$
so that the frequencies are $\pm 4.3074$, $\pm 0.6540$, and the corresponding margin in (\ref{tau*}) is
   $\tau_* = 0.7645$.  The initial covariance condition (\ref{Sigma}) is given by
   $$
    \Sigma:={\scriptsize\begin{bmatrix}
    5.9068 &  -2.2359 &  -0.8477 &   2.0721\\
   -2.2359 &   4.7534 &   4.6272 &  -2.8090\\
   -0.8477 &   4.6272 &   6.7367 &  -4.1352\\
    2.0721 &  -2.8090 &  -4.1352 &   4.8525
    \end{bmatrix}}
    $$
    and satisfies the uncertainty relation constraint $\Sigma+i\Theta \succcurlyeq 0$.
    Lemma~\ref{lem:PALE} is used in order to compute the discounted second-order moments of the system variables.
   The dependence of their real parts on the effective time horizon $\tau$   is depicted in Fig.~\ref{fig:PPP}.
\begin{figure}[thpb]
      \centering
      \includegraphics[width=85mm]{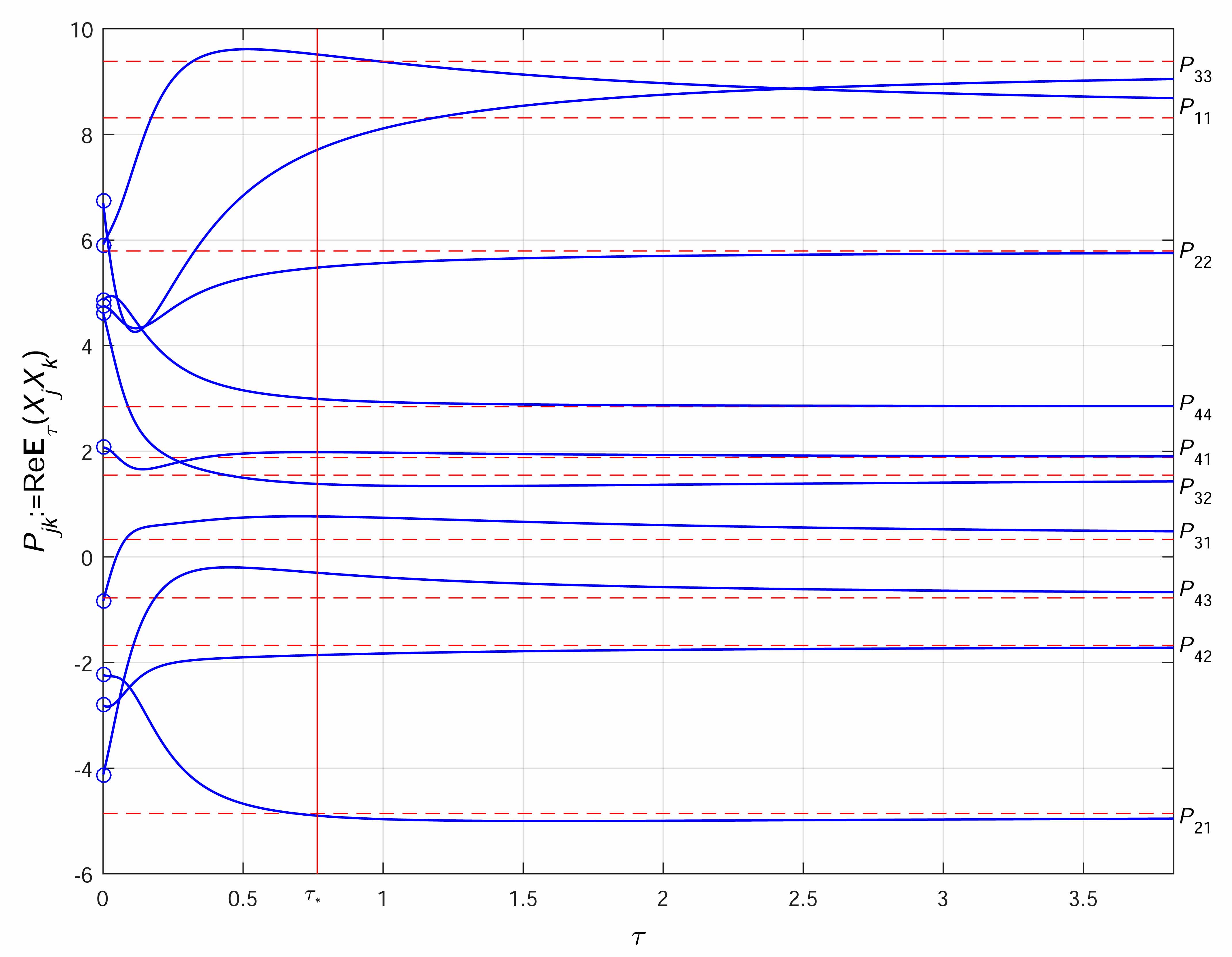}
      \caption{The dependence of the real parts $P_{jk}$ of the discounted second moments of the system variables on the effective time horizon $\tau$ for the four-dimensional QHO of Example 1.
      The dashed lines represent the corresponding infinite-horizon averages as $\tau\to +\infty$, while the ``$\circ$''s indicate the initial values $\Re\bE(X_j(0)X_k(0))$.}
      \label{fig:PPP}
   \end{figure}
These graphs show that, in the example being considered, the interval $0< \tau < 5\tau_*=3.8225$ is sufficiently large for the moments to manifest convergence to the infinite-horizon averages:
$$
    \Re \bE_{\infty}(XX^{\rT})
    =
    {\scriptsize\begin{bmatrix}
    8.3140 &  -4.8573 &   0.3322  &  1.8803\\
   -4.8573 &   5.7935 &   1.5480  & -1.6743\\
    0.3322 &   1.5480 &   9.3853  & -0.7758\\
    1.8803 &  -1.6743 &  -0.7758  &  2.8441
    \end{bmatrix}}.
$$
This matrix is calculated using the frequencies of the QHO in accordance with (\ref{EXX_red}). \hfill$\blacktriangle$

\section{Directly coupled  quantum plant and coherent quantum observer}\label{sec:QODE}

Consider a direct coupling of a quantum plant and a coherent quantum observer which form a closed QHO whose Hamiltonian $H$ is given by
\begin{equation}
\label{H}
    H
    :=
    \tfrac{1}{2}\cX^{\rT} R \cX,
    \
    \cX
    :=
    {\scriptsize\begin{bmatrix}
        X\\
        \xi
    \end{bmatrix}},
    \
    X
    :=
    {\scriptsize\begin{bmatrix}
    X_1\\
    \vdots\\
    X_n
    \end{bmatrix}},
    \
    \xi
    :=
    {\scriptsize\begin{bmatrix}
    \xi_1\\
    \vdots\\
    \xi_{\nu}
    \end{bmatrix}},
\end{equation}
where $R \in \mS_{n+\nu}$ is the plant-observer energy matrix. 
Here, $X_1, \ldots, X_n$ and $\xi_1, \ldots, \xi_{\nu}$ are the dynamic variables of the plant and the observer, respectively, with both dimensions $n$ and $\nu$ being  even.  The plant and observer variables are time-varying self-adjoint operators on the tensor-product space $\cH:= \cH_1\ox \cH_2$, where   $\cH_1$ and $\cH_2$ are initial complex separable Hilbert  spaces of the plant and the observer (which can be copies of a common Hilbert space). These quantum variables are assumed to satisfy the CCRs with a block-diagonal CCR matrix $\Theta$:
\begin{equation}
\label{Theta}
    [\cX,\cX^{\rT}]
    =
    2i \Theta,
    \qquad
    \Theta
    :=
    \diag_{k=1,2}(\Theta_k),
\end{equation}
where $\Theta_1\in \mA_n$ and $\Theta_2 \in \mA_{\nu}$ are nonsingular CCR matrices of the plant and the observer, respectively. For what follows, the plant-observer energy matrix $R$ in (\ref{H}) is partitioned as
\begin{equation}
\label{R}
    R:=
    {\scriptsize\begin{bmatrix}
        K  & L\\
        L^{\rT} & M
    \end{bmatrix}}.
\end{equation}
Here, $K\in \mS_n$ and $M\in \mS_{\nu}$ are the energy matrices of the plant and the observer which specify their free Hamiltonians $
    H_1  := \tfrac{1}{2}X^{\rT}K X
$ and $
    H_2:= \tfrac{1}{2}\xi^{\rT}M \xi$. Also, $L\in \mR^{n\x \nu}$ is the plant-observer coupling matrix which parameterizes the interaction Hamiltonian $
    H_{12}:= \tfrac{1}{2} (X^{\rT}L\xi + \xi^{\rT}L^{\rT}X) =
    \Re(X^{\rT}L\xi)
$,
where $\Re(\cdot)$ applies to operators (and matrices of operators) so that $\Re N := \tfrac{1}{2}(N+N^{\#})$ consists of self-adjoint operators. Accordingly, the total Hamiltonian $H$ in (\ref{H}) is representable as $H = H_1 + H_2 + H_{12}$. 
In view of (\ref{H})--(\ref{R}),  the Heisenberg dynamics of the composite system are governed by a linear ODE
\begin{equation}
\label{cXdot}
    \dot{\cX} = i[H,\cX] =\cA \cX.
\end{equation}
Here, in accordance with the partitioning of $\cX$ in (\ref{H}), the matrix $\cA\in \mR^{(n+\nu)\x (n+\nu)}$ is split into appropriately dimensioned blocks as
\begin{equation}
\label{cA}
    \cA
    :=
    2\Theta R
    =
    2
    {\scriptsize\begin{bmatrix}
        \Theta_1 K & \Theta_1 L\\
        \Theta_2 L^{\rT} & \Theta_2 M
    \end{bmatrix}}
    =
{\scriptsize\begin{bmatrix}
    A & BL\\
    \beta L^{\rT} & \alpha
\end{bmatrix}},
\end{equation}
with the ODE (\ref{cXdot}) being representable as a set of two ODEs
\begin{align}
\label{Xdot}
    \dot{X}  & = AX +  B \eta,\\
\label{xidot}
    \dot{\xi}  & = \alpha \xi + \beta Y,
\end{align}
where 
\begin{align}
\label{AB}
        A & := 2\Theta_1 K,
        \qquad\,
        B := 2\Theta_1,\\
\label{ab}
        \alpha & := 2\Theta_2 M,
        \qquad
        \beta := 2\Theta_2,\\
\label{Yeta}
    Y & := L^{\rT} X,
    \qquad\ \ \,
    \eta := L\xi.
\end{align}
The vector $\eta$ drives the plant variables in (\ref{Xdot}), thus resembling the classical actuator signal. The observer variables in (\ref{xidot}) are driven by the plant variables through the vector $Y$ which  corresponds to the classical  observation output from the plant.  However, the quantum mechanical nature of $Y$ and $\eta$ (which consist of time-varying self-adjoint operators on $\cH$) makes them qualitatively different from the classical signals \cite{AM_1989,KS_1972}. In view of  the relation
$
    [Y,Y^{\rT}]
    =
    L^{\rT}[X,X^{\rT}]L
    =
    2iL^{\rT}\Theta_1 L
$, following from (\ref{Theta}) and (\ref{Yeta}),  the outputs $Y_1, \ldots, Y_{\nu}$ do not commute with each other, in general, which makes them
inaccessible to simultaneous measurement. Since the plant and the observer being considered form a fully quantum system which does not involve measurements,
$Y$ is not an observation signal in the usual control theoretic sense.
In order to emphasize this distinction from the classical case, the above described observers are referred to as coherent (that is, measurement-free) quantum observers \cite{JNP_2008,L_2000,MJ_2012,NJP_2009,VP_2013b,WM_1994_paper}.
In addition to the noncommutativity of the dynamic variables, specified by the CCRs (\ref{Theta}), the quantum mechanical nature of the setting manifests itself in the fact that the ``observation'' and ``actuation'' channels  in (\ref{Yeta}) depend on the same matrix $L$. This coupling  between the ODEs (\ref{Xdot}) and (\ref{xidot}) is closely related to the Hamiltonian structure $\cA\in \Theta \mS_{n+\nu} $ of the matrix $\cA$ in (\ref{cA}).
Therefore, the ``quantum information flow'' from the plant to the observer through $Y$  has a ``back-action'' effect on the plant dynamics through $\eta$. However, unlike the conventional meaning of this term in the context of quantum measurements,  the back-action considered here is caused by the direct coupling of the observer which modifies the dynamics of the plant.

Assuming that the plant energy matrix $K$ is
fixed, 
the matrices $L$ and $M$ can be varied so as to achieve desired properties for the plant-observer QHO under constraints on the plant-observer coupling. To this end, for a given effective time horizon $\tau>0$, the observer will be called \emph{$\tau$-admissible} if the matrix $\cA$ in (\ref{cA}) satisfies
\begin{equation}
\label{taugood}
    \tau < \tfrac{1}{2\max(0,\, \ln \br(\re^{\cA}))},
\end{equation}
cf. (\ref{taumax}) of Lemma~\ref{lem:PALE}. The corresponding pairs $(L,M)$ form an open subset of $\mR^{n\x \nu}\x \mS_{\nu}$ which depends on $\tau$. In application  to the plant-observer system,  the discussions of Section~\ref{sec:aver} show that
if the matrix $R$ in (\ref{R}) is positive definite (and hence, $\cA$ has a purely imaginary spectrum), then such an observer is $\tau$-admissible for any $\tau>0$. 
The condition $R\succ 0$ is equivalent to
\begin{equation}
\label{Rpos}
    K\succ 0,
    \quad
    M \succ 0,
    \quad
    \|\Lambda\|_{\infty} <1,
    \quad
        \Lambda:=
    K^{-1/2} LM^{-1/2},
\end{equation}
where the third inequality describes the contraction property of the matrix $\Lambda$ whose largest singular value $\|\Lambda\|_{\infty}$ quantifies the ``smallness'' of the coupling matrix $L$ in comparison with the energy matrices $K$ and $M$. If the observer satisfies (\ref{Rpos}), then any system operator (with appropriate finite moments) in the plant-observer QHO lends itself to the discounted averaging, described in Section~\ref{sec:aver}, for any effective time horizon $\tau>0$. Also note that the rescaling
\begin{equation}
\label{scale}
    \wh{X}:= \sqrt{K} X,
    \qquad
    \wh{\xi}:= \sqrt{M} \xi
\end{equation}
of the plant and observer  variables
leads to a QHO with appropriately transformed CCR matrices  $\wh{\Theta}_1:= \sqrt{K}\Theta_1 \sqrt{K}$ and $\wh{\Theta}_2:= \sqrt{M}\Theta_2 \sqrt{M}$, and the energy matrix $\wh{R} := {\scriptsize\begin{bmatrix}I_n & \Lambda\\ \Lambda^{\rT} & I_{\nu}\end{bmatrix}}$, where $\Lambda$ from  (\ref{Rpos})  plays the role of the coupling matrix.

For what follows, it is assumed that the initial plant and observer variables have a block diagonal matrix of second moments:
\begin{equation}
\label{Sig}
  \Sigma
  :=
  \Re \bE(\cX_0\cX_0^{\rT})
  =
  \diag_{k=1,2}(\Sigma_k),
\end{equation}
where $\Sigma_k+i\Theta_k \succcurlyeq 0$ due to the positive semi-definiteness of quantum covariance matrices as a generalized form of the Heisenberg uncertainty principle \cite{H_2001} mentioned above. In the zero-mean case $\bE \cX_0 = 0$, this corresponds to $X_0$ and $\xi_0$ being uncorrelated. A physical rationale for the absence of initial correlation is that the observer is prepared independently of the plant and then brought into interaction with the latter at $t=0$. If the plant and the observer remained uncoupled  (which would  correspond to the case $L=0$), then, in view of Lemma~\ref{lem:PALE} and (\ref{Sig}), their variables would remain uncorrelated (in the sense that $\bE (X\xi^{\rT}) = 0$) and  the corresponding matrices $P_1 := \Re\bE_{\tau} (XX^{\rT})$ and $P_2 :=\Re \bE_{\tau}(\xi\xi^{\rT})$ would be unique solutions of independent ALEs:
\begin{align}
\label{PALE1}
    P_1 &= \tfrac{1}{\tau}\bL(A_{\tau}, \Sigma_1),
    \qquad
    A_{\tau}:= A - \tfrac{1}{2\tau}I_n,\\
\label{PALE2}
    P_2 &=\tfrac{1}{\tau}\bL(\alpha_{\tau}, \Sigma_2),
    \qquad
    \alpha_{\tau}:= \alpha - \tfrac{1}{2\tau}I_{\nu},
\end{align}
where (\ref{ILO}) is used, and both matrices $A_{\tau}$ and $\alpha_{\tau}$ are assumed to be Hurwitz.
In the general case of plant-observer coupling $L\ne 0$,  the matrix
\begin{equation}
\label{cP}
    \cP
    :=
    {\scriptsize\begin{bmatrix}
        \cP_{11} & \cP_{12}\\
        \cP_{21} & \cP_{22}
    \end{bmatrix}}
    :=
\Re \bE_{\tau} (\cX\cX^{\rT}),
\end{equation}
which is split into blocks similarly to $\cA$ in (\ref{cA}),
coincides with the controllability Gramian of the pair $(\cA_{\tau}, \sqrt{\tau^{-1}\Sigma})$ and satisfies an appropriate ALE:
\begin{align}
\label{cPALE}
    \cP = \tfrac{1}{\tau}\bL(\cA_{\tau}, \Sigma),
\end{align}
provided the observer is $\tau$-admissible in the sense of (\ref{taugood}). Here, $\Sigma$ is the initial covariance condition from (\ref{Sig}), and the matrix
 \begin{equation}
\label{cAtau}
    \cA_{\tau}
    :=
    \cA - \tfrac{1}{2\tau}I_{n+\nu}
    =
    {\scriptsize\begin{bmatrix}
        A_{\tau} & B L\\
        \beta L^{\rT} & \alpha_{\tau}
    \end{bmatrix}}
\end{equation}
is Hurwitz. As mentioned above, in the case $L=0$ (when the plant and the observer are uncoupled), the matrix $\cP$ reduces to the block diagonal matrix
\begin{equation}
\label{cP*}
  \cP_* := \diag_{k=1,2}(P_k)
\end{equation}
which is formed from the matrices $P_1$, $P_2$ in (\ref{PALE1}) and (\ref{PALE2}).

\section{Observer back-action on covariance dynamics of the plant}\label{sec:back}

The back-action of the observer can be quantified by the deviation of the covariance dynamics of the plant from those which the plant would have if it were uncoupled from the observer. In view of (\ref{cP}) and (\ref{cP*}), we will describe this deviation in terms of bilateral bounds for $\cP_{11}-P_1$ and, more generally, $\cP-\cP_*$.
To this end, we will use the following technical lemma whose proof is given here for completeness.

\begin{lem}
\label{lem:NNN}
Suppose a matrix $N \in    {\scriptsize\begin{bmatrix}
        N_{11} & N_{12}\\
        N_{21} & N_{22}
    \end{bmatrix}} \in \mS_{2n}^+$ is split into blocks $N_{jk}\in \mR^{n\x n}$. Then
\begin{equation}
\label{NNN0}
  \pm(N_{12}+N_{21})\preccurlyeq w N_{11} + \tfrac{1}{w} N_{22}
\end{equation}
for any $w>0$.    Furthermore, if $N_{11}\succ 0$ (in addition to $N\succcurlyeq 0$),  then
\begin{equation}
\label{NNN}
    \pm(N_{12}+N_{21})\preccurlyeq 2\sqrt{\br(N_{11}^{-1}N_{22})}\, N_{11}.
\end{equation}
\end{lem}
\proof
Positive semi-definiteness of the matrix $N$ implies that
$    0\preccurlyeq
    {\scriptsize\begin{bmatrix}
        \sqrt{w} I_n & \pm \tfrac{1}{\sqrt{w}} I_n
    \end{bmatrix}}
    N
    {\scriptsize\begin{bmatrix}
        \sqrt{w} I_n \\
        \pm \tfrac{1}{\sqrt{w}} I_n
    \end{bmatrix}}
    =
    w N_{11} + \tfrac{1}{w} N_{22} \pm (N_{12} + N_{21})
$
for any $w >0$, which proves (\ref{NNN0})
(similar inequalities are used, for example,  in the proof of \cite[Lemma 3]{SVP_2014}).
The parameter $w$ can be varied so as to ``tighten up'' the bound (\ref{NNN0}).
More precisely, from the additional assumption $N_{11}\succ 0$, it follows that
\begin{align}
\nonumber
    w N_{11} + \tfrac{1}{w} N_{22}
    & =
    \sqrt{N_{11}}
    \big(
        w I_n  + \tfrac{1}{w} N_{11}^{-1/2}N_{22}N_{11}^{-1/2}
    \big)
    \sqrt{N_{11}}\\
\label{NNN2}
     & \preccurlyeq
    \big(
        w + \tfrac{1}{w} \br(N_{11}^{-1}N_{22})
    \big)
    N_{11}.
\end{align}
The scalar coefficient on the right-hand side of this inequality achieves its minimum value
\begin{equation}
\label{NNN3}
    \min_{w >0}
    \big(
        w + \tfrac{1}{w} \br(N_{11}^{-1}N_{22})
    \big)
    =
    2\sqrt{\br(N_{11}^{-1}N_{22})}
\end{equation}
at $w = \sqrt{\br(N_{11}^{-1}N_{22})}$, in which case, a combination of (\ref{NNN0}), (\ref{NNN2}) and (\ref{NNN3}) leads to (\ref{NNN}).
\endproof

The following lemma will be used to give a more precise meaning to the property that the observer output  with relatively small mean square values has an appropriately weak effect on the covariance dynamics of the plant.

\begin{lem}
\label{lem:LMI}
Suppose 
the directly coupled observer is $\theta$-admissible, where
\begin{equation}
\label{varsig}
    \varsigma:= \tfrac{w\tau}{w+\tau}
    <
    \tau
    <
    \theta := \tfrac{m\tau}{m-\tau}
\end{equation}
are related to the effective time horizon $\tau$ through auxiliary parameters $w>0$ and  $m>\tau$.
Then the matrix $\cP_{11}$ from (\ref{cP}) satisfies
\begin{align}
\nonumber
    -\bL
    \big(
        A_{\varsigma},\,  \tfrac{1}{w} P_1 &+ w B\Re \bE_{\tau}(\eta\eta^{\rT})B^{\rT}
    \big)
    \preccurlyeq
    \cP_{11}-P_1\\
\label{Pdiff}
 &
     \preccurlyeq
    \bL
    \big(
        A_{\theta},\,  \tfrac{1}{m} P_1 + m B\Re \bE_{\tau}(\eta\eta^{\rT})B^{\rT}
    \big),
\end{align}
where $P_1$ is given by (\ref{PALE1}). 
\end{lem}
\proof
In view of the inequalities in (\ref{varsig}), the $\theta$-admissibility of the observer ensures that all three matrices $A_{\varsigma}$, $A_{\tau}$ and $A_{\theta}$ are Hurwitz.
Now, from the ODE (\ref{Xdot}), it follows that
$
    (XX^{\rT})^{^\centerdot}
    =
    AXX^{\rT} + XX^{\rT}A^{\rT} + B\eta X^{\rT} + X \eta^{\rT}B^{\rT}
$.
Application of the discounted averaging operator $\bE_{\tau}$ to the latter ODE and the integration by parts on its left-hand side lead to
$
    \tfrac{1}{\tau}(\cP_{11} - \Sigma_1) = A\cP_{11} + \cP_{11} A^{\rT} + \Ups
$,
and hence,
\begin{align}
\nonumber
    A_{\tau}\cP_{11}  &+ \cP_{11}A_{\tau}^{\rT} + \tfrac{1}{\tau}\Sigma_1  + \Ups\\
\label{cPZ}
     & = A_{\tau}(\cP_{11}-P_1) + (\cP_{11}-P_1)A_{\tau}^{\rT}   + \Ups = 0.
\end{align}
Here, the term
\begin{equation}
\label{Ups}
    \Ups :=
    \Re \bE_{\tau}(B\eta X^{\rT} + X \eta^{\rT}B^{\rT})
\end{equation}
originates from the plant-observer coupling and plays the role of a perturbation to the ALE $A_{\tau}
    P_1
    +
    P_1
    A_{\tau}^{\rT} + \tfrac{1}{\tau}\Sigma_1= 0
$ in (\ref{PALE1}). By 
applying the inequalities (\ref{NNN0}) of Lemma~\ref{lem:NNN} to the matrix
$
    N:= \Re\bE_{\tau} (\zeta\zeta^{\rT})= {\scriptsize\begin{bmatrix}\cP_{11} & \Re \bE_{\tau}(X \eta^{\rT})B^{\rT}\\ B\Re \bE_{\tau}(\eta X^{\rT}) & B\Re \bE_{\tau}(\eta \eta^{\rT})B^{\rT}\end{bmatrix}}\succcurlyeq 0
$
of the real parts of the second-order moments of an auxiliary vector $\zeta:= {\scriptsize\begin{bmatrix}X\\ B\eta\end{bmatrix}}$, it follows that  the matrix $\Ups$ in (\ref{Ups}) satisfies
\begin{equation}
\label{ZZ}
    \Ups\preccurlyeq \tfrac{1}{w} \cP_{11}+ w B\Re \bE_{\tau}(\eta\eta^{\rT})B^{\rT}\succcurlyeq -\Ups
\end{equation}
for any $w>0$.
Substitution of the second inequality from (\ref{ZZ}) into (\ref{cPZ}) leads to
\begin{align}
\nonumber
0  \succcurlyeq &
A_{\tau}(\cP_{11}-P_1) + (\cP_{11}-P_1)A_{\tau}^{\rT}\\
\nonumber
& -
\tfrac{1}{w} \cP_{11}- w B\Re \bE_{\tau}(\eta\eta^{\rT})B^{\rT}\\
\nonumber
=&
A_{\varsigma}(\cP_{11}-P_1) + (\cP_{11}-P_1)A_{\varsigma}^{\rT}\\
\label{LMI1}
  & -
\tfrac{1}{w} P_1- w B\Re \bE_{\tau}(\eta\eta^{\rT})B^{\rT},
\end{align}
where use is also made of the relation $\tfrac{1}{\tau} + \tfrac{1}{w} = \tfrac{1}{\varsigma}$ which follows from the definition of $\varsigma$ in (\ref{varsig}) and implies that $A_{\tau}-\tfrac{1}{2w}I_n = A_{\varsigma}$ in view of (\ref{PALE1}). The Lyapunov inequality (\ref{LMI1}) leads to the lower bound for $\cP_{11}-P_1$ in (\ref{Pdiff}). The upper bound in (\ref{Pdiff}) is established in a similar fashion by combining the ALE (\ref{cPZ}) with the first inequality from (\ref{ZZ}), except that  the parameter $m:=w$ has to satisfy $m>\tau$ in order to ensure that $\theta>0$ in (\ref{varsig}), thus making the matrix $A_{\tau}+\tfrac{1}{2m}I_n = A_{\theta}$ Hurwitz.
\endproof

The parameters $w$ and $m$ in Lemma~\ref{lem:LMI} can be varied in order to tighten up the bounds (\ref{Pdiff}), similarly to the proof of (\ref{NNN}) of Lemma~\ref{lem:NNN}. Indeed, suppose the matrix  $B\Re \bE_{\tau}(\eta\eta^{\rT})B^{\rT}$ is small in comparison with $P_1$ in terms of the dimensionless quantity
\begin{equation}
\label{small}
    \kappa:=  \tau \sqrt{\br(P_1^{-1}B\Re \bE_{\tau}(\eta\eta^{\rT})B^{\rT})},
\end{equation}
provided $P_1\succ 0$. The latter condition is fulfilled, for example, if  $\Sigma_1\succ 0$ in (\ref{Sig}).  Then, by letting $w=m=\tfrac{\tau}{\kappa}$ in   (\ref{varsig}), the corresponding $\varsigma = \tfrac{\tau}{1+\kappa}$ and $\theta = \tfrac{\tau}{1-\kappa}$ become close to $\tau$ for small values of $\kappa$, and the bounds (\ref{Pdiff}) behave asymptotically as
\begin{equation}
\label{iter}
    \pm(\cP_{11}-P_1)\precsim 2\tfrac{\kappa}{\tau}\bL(A_{\tau}, P_1)\preccurlyeq 2\kappa \br(P_1\Sigma_1^{-1})P_1.
\end{equation}
The second inequality in (\ref{iter}) is obtained by applying (\ref{LL}) to (\ref{PALE1}). Therefore, Lemma~\ref{lem:LMI} guarantees that the deviation $\cP_{11}-P_1$ is small in comparison with $P_1$ (that is, $\br(\cP_{11}P_1^{-1}-I_n)\ll 1$) and the back-action effect is negligible, if the second moments of the observer output $\eta$ are small enough in the sense that the parameter $\kappa$ in (\ref{small}) satisfies $\kappa \max(1,\br(P_1\Sigma_1^{-1}))\ll 1$.

Note that Lemma~\ref{lem:LMI} does not employ the relation (\ref{Yeta}) between the observer output $\eta$ and the observer variables $\xi$. We will therefore  provide a more accurate bound for the deviation $\cP_{11}-P_1$, which takes into account the whole plant-observer dynamics (\ref{Xdot})--(\ref{Yeta}), including the fact that $\xi$ is driven by the plant output $Y$. The formulation of the following theorem employs auxiliary matrices
\begin{align}
\label{D1}
    D_1
     & :=
    \diag(A_{\tau} \op A_{\tau}, \alpha_{\tau} \op \alpha_{\tau}),\\
\label{D2}
    D_2
     & := \diag(A_{\tau} \op \alpha_{\tau}, \alpha_{\tau} \op A_{\tau}),\\
\label{E1}
    E_1
     &:=
{\scriptsize\begin{bmatrix}
    I_n\ox (BL) & (BL)\ox I_n\\
    (\beta L^{\rT})\ox I_{\nu} & I_{\nu}\ox (\beta L^{\rT})
\end{bmatrix}},\\
\label{E2}
    E_2
     & :=
   {\scriptsize\begin{bmatrix}
       I_n\ox (\beta L^{\rT}) & (BL)\ox I_{\nu}\\
       (\beta L^{\rT})\ox I_n& I_{\nu}\ox (BL)
   \end{bmatrix}},
\end{align}
where $N\op Q:= N\ox I + I\ox Q$ is the Kronecker sum of matrices.
The matrices $D_1$ and $D_2$ are associated with the $L$-independent diagonal blocks of the matrix $\cA_{\tau}$ from (\ref{cAtau}), while $E_1$ and $E_2$ depend linearly on the coupling matrix $L$ and  are associated with the $L$-dependent off-diagonal part of $\cA_{\tau}$.

\begin{thm}
\label{th:small}
Suppose the observer is $\tau$-admissible, and both matrices $A_{\tau}$ and $\alpha_{\tau}$ in (\ref{PALE1}) and (\ref{PALE2}) are also Hurwitz. Furthermore, let the plant-observer system have the block-diagonal initial covariance condition (\ref{Sig}), and suppose the matrices
\begin{equation}
\label{Delta}
    \Delta_1 := D_1^{-1}E_1,
    \qquad
    \Delta_2:= D_2^{-1}E_2,
\end{equation}
defined in terms of (\ref{D1})--(\ref{E2}), satisfy the condition
\begin{equation}
\label{eps}
  \eps:= \|\Delta_1\|_{\infty} \|\Delta_2\|_{\infty}< 1.
\end{equation}
Then  the Frobenius norm  of the deviation of the matrix $\cP$ in (\ref{cP}) from its value $\cP_*$ for uncoupled plant and observer in (\ref{cP*}) admits upper bounds
\begin{align}
\label{PP1}
    \|\cP_{11}-P_1\|_2
    & \<
    \tfrac{\eps }{1-\eps}
    \|\cP_*\|_2,\\
\label{PP2}
    \|\cP-\cP_*\|_2
    &
 \<
\tfrac{\sqrt{1+\|\Delta_1\|_{\infty}^2}}{1-\eps}
\|\Delta_2\|_{\infty}\|\cP_*\|_2.
\end{align}
\end{thm}
\proof
Despite the symmetry of the matrix $\cP$, we will use its full (rather than half-) vectorization $\vect(\cP)\in \mR^{(n+\nu)^2}$ \cite{M_1988}. For brevity, the vectorization of a matrix will be written as $\vec{(\cdot)}$   throughout the proof. The vector $\vec{\cP}$ can be obtained by appropriately  permutating the entries of the vector ${\scriptsize\begin{bmatrix}
    \vec{\cP}_{11}\\
    \vec{\cP}_{21}\\
    \vec{\cP}_{12}\\
    \vec{\cP}_{22}
\end{bmatrix}}$. The latter satisfies the following vectorized form of the ALE (\ref{cPALE}) in view of (\ref{cAtau}):
\begin{align}
\label{cPALEvec}
{\scriptsize\begin{bmatrix}
    A_{\tau} \op A_{\tau} & I_n\ox (BL) & (BL)\ox I_n & 0\\
    I_n\ox (\beta L^{\rT}) & A_{\tau} \op \alpha_{\tau} & 0 & (BL)\ox I_{\nu}\\
    (\beta L^{\rT})\ox I_n& 0 & \alpha_{\tau} \op A_{\tau} & I_{\nu}\ox (BL)\\
    0 & (\beta L^{\rT})\ox I_{\nu} & I_{\nu}\ox (\beta L^{\rT}) & \alpha_{\tau} \op \alpha_{\tau}
\end{bmatrix}}
{\scriptsize\begin{bmatrix}
    \vec{\cP}_{11}\\
    \vec{\cP}_{21}\\
    \vec{\cP}_{12}\\
    \vec{\cP}_{22}
\end{bmatrix}}
 =
-\tfrac{1}{\tau}
{\scriptsize\begin{bmatrix}
    \vec{\Sigma}_1\\
    0\\
    0\\
    \vec{\Sigma}_2
\end{bmatrix}}.
\end{align}
The sparsity of the right-hand side of (\ref{cPALEvec}) results from the block-diagonal structure of the matrix $\Sigma$ in (\ref{Sig}) and splits the set of linear equations into the  non-homogeneous and homogeneous parts
\begin{align}
\label{cPALEvec1}
D_1
 {\scriptsize\begin{bmatrix}
    \vec{\cP}_{11}\\
    \vec{\cP}_{22}
\end{bmatrix}}
+
E_1
{\scriptsize\begin{bmatrix}
    \vec{\cP}_{21}\\
    \vec{\cP}_{12}
\end{bmatrix}}
 & =
-\tfrac{1}{\tau}
{\scriptsize\begin{bmatrix}
    \vec{\Sigma}_1\\
    \vec{\Sigma}_2
\end{bmatrix}},\\
\label{cPALEvec2}
D_2
 {\scriptsize\begin{bmatrix}
    \vec{\cP}_{21}\\
    \vec{\cP}_{12}
\end{bmatrix}}
+
E_2
{\scriptsize\begin{bmatrix}
    \vec{\cP}_{11}\\
    \vec{\cP}_{22}
\end{bmatrix}} & = 0,
\end{align}
where (\ref{D1})--(\ref{E2}) are used. With the matrices $A_{\tau}$ and $\alpha_{\tau}$ being Hurwitz, both $D_1$ and $D_2$ are nonsingular. Hence,
by solving (\ref{cPALEvec2}) for $ {\scriptsize\begin{bmatrix}
    \vec{\cP}_{21}\\
    \vec{\cP}_{12}
\end{bmatrix}}
$ and substituting the solution into (\ref{cPALEvec1}), it follows that
\begin{align}
\label{cPALEsol2}
 {\scriptsize\begin{bmatrix}
    \vec{\cP}_{21}\\
    \vec{\cP}_{12}
\end{bmatrix}}
& =
-\Delta_2
{\scriptsize\begin{bmatrix}
    \vec{\cP}_{11}\\
    \vec{\cP}_{22}
\end{bmatrix}},\\
\nonumber
 {\scriptsize\begin{bmatrix}
    \vec{\cP}_{11}\\
    \vec{\cP}_{22}
\end{bmatrix}}
& =
-\tfrac{1}{\tau}(D_1 - E_1 \Delta_2)^{-1}
{\scriptsize\begin{bmatrix}
    \vec{\Sigma}_1\\
    \vec{\Sigma}_2
\end{bmatrix}}\\
\nonumber
 & =
-\tfrac{1}{\tau}(I_{n^2+\nu^2} - \Delta_1\Delta_2)^{-1}
D_1^{-1}
{\scriptsize\begin{bmatrix}
    \vec{\Sigma}_1\\
    \vec{\Sigma}_2
\end{bmatrix}}\\
\label{cPALEsol1}
 & =
 (I_{n^2+\nu^2} - \Delta_1\Delta_2)^{-1}
{\scriptsize\begin{bmatrix}
    \vec{P_1}\\
    \vec{P_2}
\end{bmatrix}}.
\end{align}
Here, use is  made of  (\ref{Delta}) and (\ref{eps}) together with  the vectorized  representations $
    \vec{P_1} = -\tfrac{1}{\tau} (A_{\tau}\op A_{\tau})^{-1} \vec{\Sigma}_1$ and
$
    \vec{P_2} = -\tfrac{1}{\tau} (\alpha_{\tau}\op \alpha_{\tau})^{-1} \vec{\Sigma}_2
$
 for the solutions of the ALEs (\ref{PALE1}) and (\ref{PALE2}).
Since the matrix $\Delta:=\Delta_1\Delta_2$ is a contraction, with $\|\Delta\|_{\infty}\< \eps$, application of the perturbation  expansion for the matrix inverse \cite{H_2008} to (\ref{cPALEsol1}) yields
\begin{align}
\nonumber
    \left|
        {\scriptsize\begin{bmatrix}
    \vec{\cP}_{11}-\vec{P_1}\\
    \vec{\cP}_{22}-\vec{P_2}
\end{bmatrix}}
\right|
     & =
    \left|
        \Delta
        (I_{n^2+\nu^2} - \Delta)^{-1}
        {\scriptsize\begin{bmatrix}
    \vec{P_1}\\
    \vec{P_2}
    \end{bmatrix}}
    \right|
      =
    \left|
        \sum_{k=1}^{+\infty}\Delta^k
        {\scriptsize\begin{bmatrix}
    \vec{P_1}\\
    \vec{P_2}
    \end{bmatrix}}
        \right|\\
\label{PPP}
     & \<
        \sum_{k=1}^{+\infty}
        \eps^k
    \left|
        {\scriptsize\begin{bmatrix}
    \vec{P_1}\\
    \vec{P_2}
    \end{bmatrix}}
        \right|
     =
    \tfrac{\eps}{1-\eps}|\vec{\cP}_*|.
\end{align}
This implies (\ref{PP1}) since $|\vec{\cP}_*| = \sqrt{\sum_{k=1}^2\|P_k\|_2^2}=\|\cP_*\|_2$ in view of the preservation of the Frobenius norm under the vectorization. 
In order to prove (\ref{PP2}), we note that the triangle inequality  and (\ref{PPP}) lead to
\begin{equation}
\label{tri}
    \left|
            {\scriptsize\begin{bmatrix}
        \vec{\cP}_{11}\\
        \vec{\cP}_{22}
        \end{bmatrix}}
    \right|
    \<
    \left|
        {\scriptsize\begin{bmatrix}
            \vec{P_1}\\
            \vec{P_2}
        \end{bmatrix}}
    \right|
    +
    \left|
        {\scriptsize\begin{bmatrix}
            \vec{\cP}_{11}-\vec{P_1}\\
            \vec{\cP}_{22}-\vec{P_2}
        \end{bmatrix}}
    \right|
    \<
    \tfrac{\|\cP_*\|_2}{1-\eps},
\end{equation}
where the identity $    \left|
        {\scriptsize\begin{bmatrix}
            \vec{P_1}\\
            \vec{P_2}
        \end{bmatrix}}
    \right|
    =
    \|\cP_*\|_2
$ is used again. A combination of (\ref{tri}) with (\ref{cPALEsol2}) implies that
\begin{equation}
\label{off}
\left|
        {\scriptsize\begin{bmatrix}
    \vec{\cP}_{21}\\
    \vec{\cP}_{12}
    \end{bmatrix}}
\right| 
\< \|\Delta_2\|_{\infty}\tfrac{\|\cP_*\|_2}{1-\eps}.
\end{equation}
By using the orthogonal decomposition $\cP-\cP_*=\diag_{k=1,2}(\cP_{kk}-P_k)+{\scriptsize\begin{bmatrix}0 & \cP_{12}\\ \cP_{21} & 0\end{bmatrix}}$ together with (\ref{PPP}), (\ref{off}),  it follows that
$
    \|\cP-\cP_*\|_2
      =
    \sqrt{
    \left|
        {\scriptsize\begin{bmatrix}
    \vec{\cP}_{11}-\vec{P_1}\\
    \vec{\cP}_{22}-\vec{P_2}
\end{bmatrix}}
\right|^2
+
    \left|
        {\scriptsize\begin{bmatrix}
    \vec{\cP}_{21}\\
    \vec{\cP}_{12}
\end{bmatrix}}
\right|^2}
  \<
\tfrac{\sqrt{\eps^2 + \|\Delta_2\|_{\infty}^2}}{1-\eps}\|\cP_*\|_2
$,
which establishes (\ref{PP2}) in view of (\ref{eps}).
\endproof

Since Theorem~\ref{th:small} employs the standard (rather than weighted) Frobenius norm $\|\cdot\|_2$, it  would be physically more meaningful to apply the theorem to covariance dynamics of the rescaled plant and observer variables (\ref{scale}). Alternatively, Theorem~\ref{th:small} can be reformulated in terms of an appropriately weighted version of the norm. In the latter case, $\|\cP-\cP_*\|_2$ is replaced with $\|S(\cP-\cP_*)S\|_2$, where $S:= \diag(\sqrt{K}, \sqrt{M})$. We have used the standard Frobenius norm in (\ref{PP1}) and (\ref{PP2}) merely for simplicity of formulation.

There is a parallel between the proof of Theorem~\ref{th:small} and the arguments underlying the small-gain theorem (see, for example, \cite{DJ_2006} and references therein). Similar bounds for the observer back-action can be obtained in the frequency domain as outlined below. From (\ref{Xdot}), (\ref{xidot}) and (\ref{Yeta}), it follows that the Laplace transforms $\wt{X}$ and $\wt{\xi}$ of the plant and observer vectors $X$ and $\xi$, defined according to (\ref{wtX}), are related by
\begin{align}
\label{tX}
    \wt{X}(s) &= F(s) (X_0 + BL \wt{\xi}(s)),\\
\label{txi}
    \wt{\xi}(s) &= \Phi(s) (\xi_0 + \beta L^{\rT} \wt{X}(s)),
\end{align}
see Fig.~\ref{fig:small_gain}.
\begin{figure}[htpb]
\centering
\unitlength=0.9mm
\linethickness{0.2pt}
\begin{picture}(50.00,40.00)
    \put(20,30){\framebox(10,10)[cc]{{\small$F(s)$}}}
    \put(20,10){\framebox(10,10)[cc]{{\small$\Phi(s)$}}}
    \put(37,21.5){\framebox(6,6)[cc]{{\small$BL$}}}
    \put(40,27.5){\vector(0,1){4.5}}
    \put(40,15){\vector(0,1){6.5}}
    \put(30,15){\line(1,0){10}}

    \put(20,35){\line(-1,0){10}}
    \put(10,35){\vector(0,-1){7}}
    \put(10,22){\vector(0,-1){4}}
    \put(10,15){\circle{6}}
    \put(10,15){\makebox(0,0)[cc]{$+$}}
    \put(13,15){\vector(1,0){7}}
    \put(-3,15){\vector(1,0){10}}
    \put(-6,15){\makebox(0,0)[cc]{{\small$\xi_0$}}}
\put(42,15){\makebox(0,0)[lc]{{\small$\wt{\xi}(s)$}}}

\put(7,22){\framebox(6,6)[cc]{{\small$\beta\! L^{\rT}$}}}
    \put(37,35){\vector(-1,0){7}}
    \put(40,35){\circle{6}}
    \put(40,35){\makebox(0,0)[cc]{$+$}}
    \put(53,35){\vector(-1,0){10}}
    \put(55,35){\makebox(0,0)[lc]{{\small$X_0$}}}
    \put(8,35){\makebox(0,0)[rc]{{\small$\wt{X}(s)$}}}
\end{picture}\vskip-11mm
\caption{A block diagram of Eqs. (\ref{tX}) and (\ref{txi}), with the initial values $X_0$ and  $\xi_0$ shown as fictitious external inputs.  A small-gain-theorem argument applies when the coupling matrix $L$ is relatively small. 
}
\label{fig:small_gain}
\end{figure}
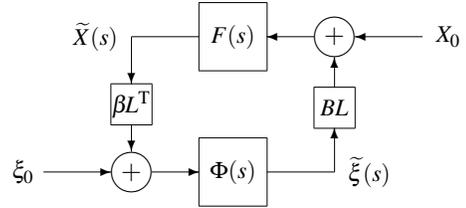
Here, $F$ and $\Phi$ are the plant and observer transfer functions, which are given by
\begin{equation}
\label{FPhi}
    F(s)
    := (sI_n-A)^{-1},
    \qquad
  \Phi(s):= (sI_{\nu}-\alpha)^{-1}
\end{equation}
in accordance with (\ref{F}) and do not depend on the coupling matrix $L$.
It follows from (\ref{tX})--(\ref{FPhi}) that
the Laplace transform of the combined vector $\cX$ of the plant and observer variables in (\ref{H}) is related to its initial value $\cX_0$ by
$
\wt{\cX}(s)
 =
 {\scriptsize\begin{bmatrix}
    \wt{X}(s)\\
    \wt{\xi}(s)
 \end{bmatrix}} = G(s) \cX_0
$
through the transfer function
\begin{equation}
\label{G}
 G(s)
 :=
 {\scriptsize\begin{bmatrix}
    I_n & -F(s)BL\\
    -\Phi(s)\beta L^{\rT} & I_{\nu}
 \end{bmatrix}}^{-1}
 {\scriptsize\begin{bmatrix}
    F(s)& 0\\
    0 & \Phi(s)
 \end{bmatrix}}.
\end{equation}
By applying (\ref{PLap}) to the plant-observer system, the matrix $\cP$ in (\ref{cP}) is represented as
\begin{equation}
\label{cPLap}
    \cP =
    \tfrac{1}{2\pi\tau}
    \Im
    \int_{\Re s=\tfrac{1}{2\tau} }
    G(s)
    (\Sigma +i\Theta)
    G(s)^*
    \rd s.
\end{equation}
The function $G$ in  (\ref{G}) differs from $\diag(F, \Phi)$ by the factor
$$
   {\scriptsize\begin{bmatrix}
    I_n & -T_1(s)\\
    -T_2(s) & I_{\nu}
 \end{bmatrix}^{\!\!-1}
 \!\!\!\!\!\!\!\!=\!\!\!\!
 \begin{bmatrix}
    (I_n-T_1(s)T_2(s))^{-1} & (I_n-T_1(s)T_2(s))^{-1}T_1(s)\\
    (I_{\nu}-T_2(s)T_1(s))^{-1}T_2(s) & (I_{\nu}-T_2(s)T_1(s))^{-1}
 \end{bmatrix}}
$$
which is close to $I_{n+\nu}$ for the relevant values of $s\in \mC$ (with $\Re s = \tfrac{1}{2\tau}$), provided the coupling matrix $L$ is small in the sense of the quantities
\begin{equation}
\label{gg}
    \gamma_k:=
    \sup_{\omega\in \mR}
    \|
        T_k(\tfrac{1}{2\tau}+i\omega)
    \|_{\infty},
    \qquad
    k = 1,2.
\end{equation}
The latter are ``discounted'' versions of the $H_{\infty}$  Hardy space norm  for
the transfer functions
$
    T_1(s):= F(s)BL$ and $
    T_2(s):= \Phi(s) \beta L^{\rT}
$
which depend linearly on $L$.
Therefore,
the frequency-domain representation  (\ref{cPLap}) can be used together with the parameters (\ref{gg}) in order to obtain bounds for the deviation of $\cP$ from the matrix
$
    \cP_*
    =
    \tfrac{1}{2\pi\tau}
    \Im
    \int_{\Re s = \tfrac{1}{2\tau}}
 \diag(F(s)\Gamma_1F(s)^*, \Phi(s)\Gamma_2\Phi(s)^*)
     \rd s
$
under the small-gain condition $\gamma_1\gamma_2<1$ for the loop in Fig.~\ref{fig:small_gain}.
Here, $\Gamma_k:= \Sigma_k +i\Theta_k$, with $k=1,2$, are the initial second moment matrices for the plant and observer variables, which, in accordance with (\ref{Sig}),  form the corresponding matrix for the closed-loop system:     $\Gamma
    :=
    \bE (\cX_0 \cX_0^\rT)
    =
    \diag_{k=1,2}(\Gamma_k)$.

\section{Discounted mean square optimal coherent quantum filtering problem}\label{sec:CQF}

For what follows, let the plant energy matrix $K$ be fixed and satisfy $K\succ 0$. The latter is sufficient (but not necessary) for the set of $\tau$-admissible observers to be nonempty for any given $\tau>0$. In particular, this set contains observers with $L=0$ and arbitrary $M\succ 0$ (in which case $R=\diag(K,M)\succ 0$), or more generally, observers satisfying (\ref{Rpos}). 
Consider a CQF problem
\begin{equation}
\label{cZ}
    \cZ
    :=
    \bE_{\tau} Z
        \longrightarrow \min
\end{equation}
of minimizing  a quadratic cost functional
over the plant-observer coupling matrix $L$  and the observer energy matrix $M$ subject to the constraint (\ref{taugood}).
Here, $\tau>0$ is a given effective time horizon for the discounted averaging (\ref{bEtau}) which is applied to the quantum process
\begin{equation}
    \label{ZE}
    Z
    :=
    E^{\rT}E
    +
    \lambda
    \eta^{\rT} \Pi \eta
    =
    \cX^{\rT}
    \cC^{\rT}\cC
    \cX.
\end{equation}
The latter is a time-varying self-adjoint operator on the plant-observer space $\cH$ which is defined in terms of
the vectors $\cX$, $\eta$ from  (\ref{H}), (\ref{Yeta}), and
\begin{equation}
\label{ES}
    E
    :=
    S_1X-S_2\xi
    =
    S\cX,
    \
    S
    :=
    {\scriptsize\begin{bmatrix}
        S_1 & -S_2
    \end{bmatrix}},
    \
    \cC:= {\scriptsize\begin{bmatrix} S_1 & -S_2\\ 0 & \sqrt{\lambda \Pi} L\end{bmatrix}}.
\end{equation}
Here,
$S_1\in \mR^{p\x n}$, $S_2\in \mR^{p\x \nu}$ and $\Pi \in \mS_n$ are given matrices, with $\Pi\succ 0$,  which, together with a given scalar parameter $\lambda>0$, determine the matrix $\cC\in \mR^{(p+n)\x (n+\nu)}$ (with the first block-row $S \in \mR^{p \x (n+\nu)}$) and its dependence on the coupling matrix $L$.
 The matrix  $S_1$ specifies linear combinations of the plant variables of interest which are to be approximated by given linear functions of the observer variables specified by the matrix $S_2$. Accordingly, the vector $E$ in (\ref{ES}) (consisting of $p$ time-varying self-adjoint operators on $\cH$)
is interpreted as an \emph{estimation error}. In addition to the discounted mean square $\bE_{\tau} (E^{\rT} E)$ of the estimation error,  the cost functional $\cZ$ in (\ref{cZ}) involves a quadratic penalty $\bE_{\tau}(\eta^{\rT} \Pi\eta)$ for the observer back-action on the covariance dynamics of the plant (see Lemma~\ref{lem:LMI}), with $\lambda$ being the relative weight of this penalty in $\cZ$. In fact, $\cZ$ is organised as the Lagrange function for a related CQF problem of minimizing the discounted mean square of the estimation error subject to an additional weighted mean square constraint on the plant-observer coupling:
\begin{equation}
\label{CQF2}
    \bE_{\tau} (E^{\rT} E)
    \longrightarrow
    \min,
    \qquad
    \bE_{\tau} (\eta^{\rT}\Pi\eta)\< r.
\end{equation}
In this formulation,  $\lambda$ plays the role of a Lagrange multiplier which is found so as to make the solution of (\ref{cZ}) saturate the constraint in (\ref{CQF2}) for a given threshold $r> 0$.
In a particular case  $S_2=0$, the CQF problem  (\ref{cZ})--(\ref{ES}) is a quantum mechanical analogue of the LQR problem \cite{AM_1989,KS_1972} in view of the analogy between the observer output $\eta$ and classical actuation signals discussed in Section~\ref{sec:QODE}.  The presence of the  quantum expectation  of a nonlinear function of system variables in (\ref{cZ}) and the optimization requirement make this setting different from the time-averaged approach of \cite{P_2014,PH_2015}.

Substitution of  (\ref{ZE}) into (\ref{cZ}) allows the cost functional to be  expressed in terms of the matrix $\cP$ from (\ref{cP}) as
\begin{equation}\label{ZP}
  \cZ = \bra \cC^{\rT}\cC, \bE_{\tau}(\cX\cX^{\rT})\ket = \bra \cC^{\rT}\cC, \cP\ket,
\end{equation}
where $\bra\cdot, \cdot\ket$ is the Frobenius inner product of matrices. Under the assumptions of Theorem~\ref{th:small},  a combination of (\ref{ZP}) with the Cauchy-Bunyakovsky-Schwarz inequality and the bound (\ref{PP2}) leads to
\begin{align}
\nonumber
    |\bE_{\tau}&(E^{\rT}E) - \bra S^{\rT}S, \cP_*\ket|
      = |\bra S^{\rT}S, \cP-\cP_*\ket|\\
\label{CEI}
    \< &
    \|S^{\rT}S\|_2\,
    \|\cP-\cP_*\|_2
    \<
    \|S^{\rT}S\|_2
    \tfrac{\sqrt{1+\|\Delta_1\|_{\infty}^2}}{1-\eps}\|\Delta_2\|_{\infty}\|\cP_*\|_2,
\end{align}
which relates the discounted mean square of the estimation error with the plant-observer coupling strength quantified by $\|\Delta_1\|_{\infty}$, $\|\Delta_2\|_{\infty}$ and $\eps$ from (\ref{Delta}) and (\ref{eps}).
Here, $S$ is
the first block-row of the matrix $\cC$ in (\ref{ES}), so that
$
    \|S^{\rT}S\|_2
    =
    \|SS^{\rT}\|_2
    =
    \sqrt{\Tr((\sum_{k=1}^2 S_kS_k^{\rT})^2)}
$, and $\bra S^{\rT}S, \cP_*\ket = \sum_{k=1}^2\Tr (S_k P_k S_k^{\rT})$ in view of the block-diagonal structure of the matrix $\cP_*$ in (\ref{cP*}). Therefore, the inequality (\ref{CEI}) implies that
\begin{align}
\nonumber
    \bE_{\tau}(E^{\rT}E)
     \> &     \sum_{k=1}^2\Tr (S_k P_k S_k^{\rT})\\
\label{lower}
       -&
    \sqrt{\Tr\Big(\Big(\sum_{k=1}^2 S_kS_k^{\rT}\Big)^2\Big)}
    \tfrac{\sqrt{1+\|\Delta_1\|_{\infty}^2}}{1-\eps}\|\Delta_2\|_{\infty}\|\cP_*\|_2,
\end{align}
which becomes an equality if $L=0$. The right-hand side of (\ref{lower}) provides a lower bound for the mean square of the estimation error. This bound depends on the coupling matrix $L$ only through $\|\Delta_1\|_{\infty}$, $\|\Delta_2\|_{\infty}$, $\eps$ and shows that $L$ has to be sufficiently large in order to make $\bE_{\tau}(E^{\rT}E)$ smaller than $\sum_{k=1}^2\Tr (S_k P_k S_k^{\rT})$ by a given amount. At the same time, the plant-observer coupling should be weak enough to avoid severe back-action of the observer on the plant. Therefore, the parameter $\lambda$ in the CQF problem (\ref{cZ})--(\ref{ES}) quantifies a compromise  between these  conflicting requirements (of minimizing the estimation error and reducing the back-action).

\section{First-order    necessary conditions of optimality}\label{sec:opt}

The following theorem provides first-order necessary conditions of optimality for the CQF problem (\ref{cZ})--(\ref{ES}).  Their formulation employs
the Hankelian
\begin{equation}
\label{cE}
    \cE
    :=
    {\scriptsize\begin{bmatrix}
        \cE_{11} & \cE_{12}\\
        \cE_{21} & \cE_{22}
    \end{bmatrix}}
    :=
    \cQ\cP,
\end{equation}
associated with the matrix $\cP$ from (\ref{cP}) and the observability Gramian $\cQ$ of $(\cA_{\tau}, \cC)$ which is a unique solution of the corresponding ALE:
\begin{align}
\label{cQALE}
    \cQ :=
        {\scriptsize\begin{bmatrix}
        \cQ_{11} & \cQ_{12}\\
        \cQ_{21} & \cQ_{22}
    \end{bmatrix}}
    =
    \bL(\cA_{\tau}^{\rT}, \cC^{\rT}\cC).
\end{align}
 The matrices $\cE$ and $\cQ$ are split into appropriately dimensioned blocks $(\cdot)_{jk}$ similarly to the matrix $\cP$ in (\ref{cP}), with $(\cdot)_{j \bullet}$ the $j$th block-row  and  $(\cdot)_{\bullet k}$  the $k$th block-column of the  matrices.

\begin{thm}
\label{th:stat}
Suppose the plant energy matrix satisfies $K\succ 0$, and the directly coupled observer is $\tau$-admissible in the sense of (\ref{taugood}). Then the observer is a stationary point of the CQF problem (\ref{cZ})--(\ref{ES}) if and only if the Hankelian $\cE$ in (\ref{cE}) and the controllability Gramian $\cP$ in (\ref{cP}) satisfy
\begin{align}
\label{dcZdL0}
    \Theta_1 \cE_{12}
    -\cE_{21}^{\rT}\Theta_2
    & =\tfrac{\lambda}{2} \Pi L \cP_{22},\\
\label{dcZdM0}
    \Theta_2 \cE_{22}-\cE_{22}^{\rT}\Theta_2& =0.
\end{align}
\end{thm}
\proof
By using (\ref{cPALE}) and the duality $\bL(\cA_{\tau}, \cdot)^{\dagger} = \bL(\cA_{\tau}^{\rT}, \cdot)$, it follows that the cost $\cZ$ in (\ref{ZP}) is representable in terms of the observability Gramian $\cQ$ from (\ref{cQALE}) as
\begin{equation}
\label{ZQ}
    \cZ
    =
    \tfrac{1}{\tau}
    \Bra
        \cC^{\rT}\cC,
        \bL(\cA_{\tau}, \Sigma)
    \Ket
    =
    \tfrac{1}{\tau}
    \Bra
        \bL(\cA_{\tau}^{\rT}, \cC^{\rT}\cC),
        \Sigma
    \Ket
    =
    \tfrac{1}{\tau}
    \bra \cQ, \Sigma\ket.
\end{equation}
Here, the adjoint $(\cdot)^{\dagger}$ of linear operators on matrices is in the sense of the Frobenius inner product.
With the matrix $\cA_{\tau}$ in (\ref{cAtau}) being Hurwitz due to the $\tau$-admissibility constraint (\ref{taugood}), the representation (\ref{ZQ}) shows that $\cZ$ inherits  a smooth dependence on  $L$ and $M$ from $\cQ$. The latter is a composite function $(L,M)\mapsto (\cA, \cC)\mapsto \cQ$ whose first variation is
\begin{equation}
\label{dcQ}
    \delta \cQ = \bL(\cA_{\tau}^{\rT},\, (\delta \cA)^{\rT} \cQ + \cQ \delta \cA + (\delta \cC)^{\rT}\cC + \cC^{\rT} \delta \cC),
\end{equation}
where use is made of the ALE in (\ref{cQALE}), and  the first variations of the matrices $\cA$ in (\ref{cA}) and $\cC$ in (\ref{ES}) with respect to $L$ and $M$ are
\begin{equation}
\label{dcAdcC}
    \delta\cA
    =
    2\Theta
    {\scriptsize\begin{bmatrix}
        0 & \delta L \\
        \delta L^{\rT} & \delta M
    \end{bmatrix}},
    \qquad
    \delta\cC
    =
    {\scriptsize\begin{bmatrix}
        0 & 0\\
        0 & \sqrt{\lambda \Pi}\delta L
    \end{bmatrix}}.
\end{equation}
By combining the duality argument above with (\ref{dcQ}) and (\ref{dcAdcC}), it follows that the first variation of $\cZ$ in (\ref{ZQ}) can be computed as
\begin{align}
\nonumber
    \delta\cZ
    =&
    \tfrac{1}{\tau}
    \Bra
        \bL(\cA_{\tau}^{\rT}, (\delta \cA)^{\rT} \cQ + \cQ \delta \cA + (\delta \cC)^{\rT}\cC + \cC^{\rT} \delta \cC),
        \Sigma
    \Ket\\
\nonumber
    =&
    \Bra
        (\delta \cA)^{\rT} \cQ + \cQ \delta \cA + (\delta \cC)^{\rT}\cC + \cC^{\rT} \delta \cC,
        \cP
    \Ket\\
\nonumber
    =&
    2\Bra
        \cE, \delta \cA
    \Ket
    +
    2
    \Bra
        \cC\cP, \delta \cC
    \Ket    \\
\nonumber
    =&
    -
    4
    \Bra
     \Theta \cE,
     { \scriptsize\begin{bmatrix} 0 & \delta L \\ \delta L^{\rT} & \delta M\end{bmatrix}}
    \Ket
    +
    2 \Bra \cC\cP, { \scriptsize\begin{bmatrix} 0 & 0\\ 0 & \sqrt{\lambda \Pi} \delta L\end{bmatrix}}\Ket\\
\nonumber
    =&
    -
    4
    \Bra
     \bS(\Theta \cE),
     { \scriptsize\begin{bmatrix} 0 & \delta L \\ \delta L^{\rT} & \delta M\end{bmatrix}}
    \Ket
    +
    2 \Bra (\cC\cP)_{22}, \sqrt{\lambda \Pi}\delta L\Ket\\
\nonumber
    =&
    -8\Bra
        \bS(\Theta \cE)_{12},
        \delta L
    \Ket
    -4\Bra
        \bS(\Theta \cE)_{22},
        \delta M
    \Ket\\
\nonumber
    & +
2 \bra \sqrt{\lambda \Pi} L\cP_{22}, \sqrt{\lambda \Pi}\delta L\ket\\
\label{delcZ}
    =&
2 \Bra \lambda \Pi L\cP_{22} - 4\bS(\Theta \cE)_{12}, \delta L\Ket
     -
    4
    \Bra
     \bS(\Theta \cE)_{22},
     \delta M
    \Ket
\end{align}
(similar calculations can be found, for example, in \cite{VP_2013a}). Here,  $\bS(N):= \tfrac{1}{2}(N+N^{\rT})$ denotes the symmetrizer of matrices, so that
\begin{equation}
\label{bSTE}
    \bS(\Theta \cE)
     =
    \tfrac{1}{2}
    (\Theta \cE - \cE^{\rT}\Theta)
    =
    \tfrac{1}{2}
    {\scriptsize\begin{bmatrix}
        \Theta_1 \cE_{11}-\cE_{11}^{\rT}\Theta_1  &
        \Theta_1 \cE_{12}-\cE_{21}^{\rT}\Theta_2\\
        \Theta_2 \cE_{21}-\cE_{12}^{\rT}\Theta_1 &
        \Theta_2 \cE_{22}-\cE_{22}^{\rT}\Theta_2\end{bmatrix}}.
\end{equation}
A combination of (\ref{delcZ}) with (\ref{bSTE}) leads to the  partial Frechet derivatives of $\cZ$ on the corresponding Hilbert spaces of matrices $\mR^{n\x \nu}$ and $\mS_{\nu}$:
\begin{align}
\nonumber
    \d_L \cZ
    & =
    2(\lambda \Pi L\cP_{22} - 4\bS(\Theta \cE)_{12})\\
\label{dcZdL}
     & = 2 (\lambda \Pi L\cP_{22} -
    2
    (\Theta_1 \cE_{12}-\cE_{21}^{\rT}\Theta_2)),\\
\label{dcZdM}
    \d_M \cZ
    & = -4\bS(\Theta_2 \cE_{22})
     = -2(\Theta_2 \cE_{22}-\cE_{22}^{\rT}\Theta_2).
\end{align}
By equating the Frechet derivatives (\ref{dcZdL}) and (\ref{dcZdM}) to zero, it follows that the stationarity of $\cZ$ with respect to $L$ and $M$ is equivalent to (\ref{dcZdL0}) and (\ref{dcZdM0}). \endproof

The relation (\ref{bSTE}) implies that  the fulfillment of the first-order optimality conditions (\ref{dcZdL0}) and (\ref{dcZdM0}) for the observer is equivalent to the existence of a matrix $N\in \mS_n$ such that
\begin{equation}
\label{STE}
    \Theta \cE - \cE^{\rT} \Theta
    =
    \tfrac{1}{2}
    {\scriptsize\begin{bmatrix}
    N & \lambda \Pi L\cP_{22}\\
    \lambda \cP_{22}L^{\rT}\Pi & 0
    \end{bmatrix}}.
\end{equation}
Here, the zero block corresponds to (\ref{dcZdM0}), which means that the matrix $\cE_{22}$ is skew-Hamiltonian in the sense of the symplectic structure specified by $\Theta_2^{-1}$, that is, $\cE_{22} \in \Theta_2^{-1} \mA_{\nu}$.

A quantum probabilistic interpretation of the optimality conditions (\ref{dcZdL0}) and (\ref{dcZdM0}) is that, for any such observer, the process $\vartheta$, given by
\begin{equation}
\label{vart}
    \vartheta
    :=
    {\scriptsize
    \begin{bmatrix}
      \vartheta_1\\
      \vartheta_2
    \end{bmatrix}}
    :=
    \Theta
    \cQ \cX,
    \quad
    \vartheta_j
    :=
    \Theta_j \cQ_{j\bullet} \cX,
    \quad
    j=1,2,
\end{equation}
and consisting of $n+\nu$ self-adjoint operators (which are special linear combinations of the plant and observer variables),
satisfies the covariance relations
\begin{align}
\nonumber
  \bE_\tau(\vartheta &\xi^{\rT} + \cX \vartheta_2^{\rT})
  =
  \Theta \cQ\bE_\tau(\cX \xi^{\rT})
  -
    \bE_\tau(\cX \cX^{\rT}) \cQ_{\bullet 2}\Theta_2\\
\nonumber
  = &
  \Theta \cQ
  {\scriptsize\begin{bmatrix}
    \cP_{12}\\
    \cP_{22} + i\Theta_2
  \end{bmatrix}}
  -
  (\cP + i\Theta)
  \cQ_{\bullet 2}\Theta_2
  =
    \Theta \cE_{\bullet 2} - \cE_{2\bullet}^\rT\Theta_2\\
\label{covrel}
    =&
    {\scriptsize
    \begin{bmatrix}
        \tfrac{\lambda}{2} \Pi L \cP_{22}\\
        0
    \end{bmatrix}}
    =
    {\scriptsize
    \begin{bmatrix}
        \tfrac{\lambda}{2} \Pi \Re \bE_\tau (\eta \xi^\rT)\\
        0
    \end{bmatrix}}.
\end{align}
Here, use is made of the identities $\cE_{jk} = \cQ_{j\bullet} \cP_{\bullet k}$ and $\cE_{jk}^{\rT} = \cP_{k\bullet} \cQ_{\bullet j}$, which follow from (\ref{cE}) and the symmetry of the Gramians $\cP$ and $\cQ$ in (\ref{cP}) and (\ref{cQALE}). In particular, (\ref{covrel}) implies that $\vartheta_2$ in (\ref{vart}) and $\xi$ are uncorrelated in the sense that
\begin{equation}
\label{uncorr}
    \bE_\tau(\vartheta_2 \xi^{\rT} + \xi\vartheta_2^{\rT}) = 0.
\end{equation}
This is a quantum counterpart of the corresponding property for the state estimation error and the state estimate in the classical Kalman filter \cite{AM_1979}.

If $\cP_{22}\succ 0$, then, in view of the assumption $\Pi\succ 0$,  (\ref{dcZdL0}) implies that the optimal coupling matrix is representable as
\begin{equation}
\label{L}
    L
    =
    \tfrac{2}{\lambda}
    \Pi^{-1}
    (\Theta_1 \cE_{12}
    -\cE_{21}^{\rT}\Theta_2)
    \cP_{22}^{-1}.
\end{equation}
In order to close the ALEs (\ref{cPALE}) and (\ref{cQALE}), the relation (\ref{L}) 
needs to be complemented with an appropriate equation for the optimal observer matrix $M$. The latter step is less straightforward and will be considered in the next section.

\section{Lie-algebraic representation of optimality conditions}\label{sec:lie}

For what follows, we associate with the Gramians $\cP$ and $\cQ$ from (\ref{cP}) and (\ref{cQALE}) the matrices
\begin{equation}
\label{PQ}
    P:= \cP \Theta^{-1},
    \qquad
    Q:= \Theta \cQ
\end{equation}
belonging to the same subspace $\Theta \mS_{n+\nu}$ of Hamiltonian matrices as $\cA$ in (\ref{cA}). Here, the property $P \in \Theta \mS_{n+\nu}$ follows from $\Theta^{-1}\cP \Theta^{-1}\in \mS_{n+\nu}$.  The linear space $\Theta \mS_{n+\nu}$, equipped with the commutator $[\cdot, \cdot]$, 
 is a Lie algebra \cite{D_2006,O_1993,V_1984}, in terms of which the ALEs and the optimality conditions above will be reformulated by the following lemma. Its formulation employs the Hamiltonian matrix
\begin{equation}
\label{D}
  D
  :=
  [Q,P]
 =
  \Theta \cQ \cP\Theta^{-1}
  -
  \cP \cQ
  =
    (\Theta \cE
  -
  \cE^\rT \Theta)
  \Theta^{-1},
\end{equation}
which (for any $\tau$-admissible observer) is related to the left-hand side of (\ref{STE}) due to
(\ref{cE}), (\ref{PQ}) and the symmetry of the Gramians $\cP$, $\cQ$.

\begin{lem}
\label{lem:lie}
The ALEs (\ref{cPALE}), (\ref{cQALE}) and the optimality conditions (\ref{dcZdL0}), (\ref{dcZdM0}) for the CQF problem (\ref{cZ})--(\ref{ES}) are representable in a Lie-algebraic form through the Hamiltonian matrices $P$, $Q$ from (\ref{PQ}):
\begin{align}
\label{cPQALElie1}
    [\cA, P] &= \tfrac{1}{\tau}(P-\Sigma\Theta^{-1}),\\
\label{cPQALElie2}
    [\cA, Q] & = \Theta \cC^{\rT}\cC - \tfrac{1}{\tau} Q,\\
\label{optlie1}
    D_{12}
    & =
    \tfrac{\lambda}{2}
    \Pi LP_{22},\\
\label{optlie2}
    D_{22}
    & =
    0,
\end{align}
where $D_{12}$ and $D_{22}$ are the corresponding blocks of the matrix $D\in \Theta  \mS_{n+\nu}$ in (\ref{D}).
\end{lem}
\proof
The Hamiltonian structure of the matrix $\cA$  in (\ref{cA}) implies that
$    \cA^{\rT} = -\Theta^{-1}\cA \Theta
$,
and hence, 
\begin{align}
\nonumber
    \cA_{\tau}\cP & + \cP \cA_{\tau}^{\rT}
      =
    \cA\cP + \cP \cA^{\rT}     - \tfrac{1}{\tau}\cP\\
\label{AAPP}
      = &
    \cA\cP - \cP \Theta^{-1}\cA \Theta     - \tfrac{1}{\tau}\cP
    =
    \big(
        [\cA, P] - \tfrac{1}{\tau}P
    \big)\Theta,\\
\nonumber
    \cA_{\tau}^{\rT}\cQ  & + \cQ \cA_{\tau}
      =
    \cA^{\rT}\cQ + \cQ \cA     - \tfrac{1}{\tau}\cQ\\
\label{AAQQ}
       = &
    -\Theta^{-1}\cA \Theta\cQ +\cQ \cA - \tfrac{1}{\tau}\cQ
    =
    -\Theta^{-1}
    \big([\cA, Q] + \tfrac{1}{\tau}Q\big),
\end{align}
where use is also made of (\ref{cAtau}) and (\ref{PQ}). Substitution of (\ref{AAPP}) and (\ref{AAQQ}) into the ALEs (\ref{cPALE}), (\ref{cQALE})
leads to their Lie-algebraic representations (\ref{cPQALElie1}), (\ref{cPQALElie2}). 
Furthermore, by substituting (\ref{STE}) into (\ref{D}), considering the second block-column
$
    D_{\bullet 2}
    =
    \tfrac{1}{2}
    {\scriptsize\begin{bmatrix}
    \lambda \Pi L\cP_{22}\\
    0
    \end{bmatrix}}
    \Theta_2^{-1}
$
and using the relation $\cP_{22}\Theta_2^{-1} = P_{22}$, it follows that 
the optimality conditions (\ref{dcZdL0}) and (\ref{dcZdM0}) admit the Lie-algebraic representations (\ref{optlie1}) and (\ref{optlie2}). 
\endproof

The solutions of (\ref{cPQALElie1}) and (\ref{cPQALElie2}) admit the representation 
\begin{align}
\label{Pres}
    P    & = (\cI - \tau \ad_{\cA})^{-1}(\Sigma\Theta^{-1}),\\
\label{Qres}
    Q    & = \tau (\cI + \tau \ad_{\cA})^{-1}(\Theta \cC^{\rT}\cC),
\end{align}
where $\cI$ is the identity operator on the space $\Theta \mS_{n+\nu}$. 
Here, the resolvents  $(\cI \pm \tau \ad_{\cA})^{-1}$ are well-defined since the $\tau$-admissibility (\ref{taugood}) implies that  the spectrum of the linear operator $\ad_{\cA}$ on $\Theta \mS_{n+\nu}$ is contained in the strip $\{z \in \mC:\ |\Re z|< \tfrac{1}{\tau}\}$. 

\begin{lem}
\label{lem:LM}
The 
optimal coupling matrix $L$ in (\ref{L}) can be expressed in terms of the matrices $P$ and $Q$ from (\ref{PQ}) as
\begin{align}
\label{Lopt}
    L
    & =
    \tfrac{2}{\lambda}
    \Pi^{-1}
    D_{12}
    P_{22}^{-1},
\end{align}
provided $\cP_{22}\succ 0$, where the matrix $D$ is given by (\ref{D}). Furthermore, the optimal energy matrix $M$ of the observer satisfies
\begin{align}
\nonumber
    \tfrac{1}{2}\Big(\tfrac{1}{\tau}[\Sigma \Theta^{-1}, Q]_{12}&+[\Theta \cC^{\rT}\cC, P]_{12}\Big)\\
\label{eq12}
    + D_{11}\Theta_1 L &-
    \Theta_1
    KD_{12}
 + D_{12}\Theta_2 M = 0.
\end{align}
\end{lem}
\proof
The representation 
(\ref{Lopt}) follows directly from the first optimality condition (\ref{optlie1}) under the assumption $\cP_{22}\succ 0$. 
In order to establish (\ref{eq12}), we note that
the  left-hand sides of (\ref{cPQALElie1})--(\ref{optlie2}) involve pairwise commutators of the Hamiltonian matrices $\cA, P, Q \in \Theta \mS_{n+\nu}$. Application of the Jacobi identity \cite{V_1984} and the antisymmetry of the commutator leads to the relations
\begin{align}
\nonumber
0   & =
[[P,\cA],Q] +[[\cA,Q], P]+[[Q,P], \cA]\\
\nonumber
& =
\tfrac{1}{\tau}
[\Sigma\Theta^{-1}-P,Q]
+
\big[\Theta \cC^{\rT}\cC-\tfrac{1}{\tau} Q, P\big]
 +[D, \cA]\\
\label{jacob}
& =  \tfrac{1}{\tau}
[\Sigma\Theta^{-1},Q]
+
[\Theta \cC^{\rT}\cC, P]+[D, \cA]
\end{align}
for any $\tau$-admissible observer, where use is made of (\ref{D}) (here, neither of the optimality conditions (\ref{optlie1}) and (\ref{optlie2}) has been used). By substituting the matrix $\cA$ from (\ref{cA}) into the right-hand side of (\ref{jacob}) and considering the $(\cdot)_{12}$ block of the resulting Hamiltonian matrix, it follows that  \begin{align}
\nonumber
    \tfrac{1}{\tau}[&\Sigma \Theta^{-1}, Q]_{12}+[\Theta \cC^{\rT}\cC, P]_{12}\\
\label{jacob12}
    +& 2\big(D_{11}\Theta_1 L +D_{12}\Theta_2 M
    -
    \Theta_1 (
    KD_{12}+LD_{22})\big) = 0.
\end{align}
Now, the second optimality condition  (\ref{optlie2}) makes the corresponding term in (\ref{jacob12}) vanish, thus leading to (\ref{eq12}). 
\endproof

As can be seen from the proof of Lemma~\ref{lem:LM}, the relation (\ref{eq12}) holds for any $\tau$-admissible stationary point of the CQF problem regardless of the assumption $\cP_{22}\succ 0$.   Furthermore, (\ref{eq12}) is a linear equation with respect to $M$. This allows the optimal observer energy matrix  $M$ to be  expressed in terms of $P$, $Q$ from (\ref{PQ}) in the case of equal plant and observer dimensions $n=\nu$. 
In this case, the  observer will be called \emph{nondegenerate} if the matrices  $P$ and  $D$  from (\ref{PQ}) and (\ref{D})   satisfy
\begin{equation}
\label{nondeg}
  \cP_{22}\succ 0,
  \qquad
  \det D_{12}\ne 0.
\end{equation}
The above results 
lead to the following necessary conditions of optimality for such observers.

\begin{thm}
\label{th:LM}
Suppose the plant and observer dimensions are equal: $n=\nu$. Then for any nondegenerate observer, which  is a stationary point of the CQF problem (\ref{cZ})--(\ref{ES}) under the assumptions of Theorem~\ref{th:stat}, the coupling and energy matrices are related by (\ref{Lopt}) and
\begin{align}
\nonumber
    \!M = &
    \Theta_2^{-1}D_{12}^{-1}\big(\Theta_1 KD_{12}\\
\label{Mopt}
    &  - D_{11}\Theta_1 L- \tfrac{1}{2}\big(\tfrac{1}{\tau}[\Sigma \Theta^{-1}, Q]_{12}+[\Theta \cC^{\rT}\cC, P]_{12}\big)\big)
\end{align}
to the matrices $P$ and $Q$ from (\ref{PQ}) satisfying the ALEs (\ref{cPQALElie1}) and (\ref{cPQALElie2}).
\end{thm}
\proof
The first of the conditions (\ref{nondeg}) makes the  representation (\ref{Lopt}) applicable, which leads to a nonsingular coupling matrix $L$ in view of the second condition in (\ref{nondeg}). The latter allows (\ref{eq12}) to be 
uniquely solved for the observer energy matrix $M$ in the form (\ref{Mopt}). 
\endproof

The first line of (\ref{Mopt}) is organised as a similarity transformation which would relate the Hamiltonian matrices $\Theta_1 K$ and $\Theta_2 M$ if there were no additional terms on the right-hand side of the equation. 
In that case, the transformation matrix $D_{12}$ in (\ref{Mopt}) would preserve the Hamiltonian structure if  it were symplectic in the generalized sense that $D_{12}\Theta_2 D_{12}^{\rT} = \Theta_1$.

In combination with the ALEs (\ref{cPALE}) and (\ref{cQALE}) (or their Lie-algebraic form (\ref{PQ})--(\ref{cPQALElie2}), the relations (\ref{Lopt}) and (\ref{Mopt}) of Lemma~\ref{lem:LM} and Theorem~\ref{th:LM} provide a set of algebraic equations for finding the matrices $L$ and $M$ of a nondegenerate observer among stationary points in the CQF problem (\ref{cZ})--(\ref{ES}).

\section{Observers with autonomous estimation error dynamics}
\label{sec:sub}

In view of the complicated structure of the equations of Sections~\ref{sec:opt} and \ref{sec:lie} 
for an optimal observer, consider a suboptimal solution of the CQF problem in a special class of observers which lead to autonomous dynamics of the estimation error $E$   in (\ref{ES}). More precisely, suppose the observer is such that
\begin{equation}
\label{SAAS}
    S \cA = \wh{\cA} S
\end{equation}
for some $\wh{\cA} \in \mR^{p\x p}$,
where the matrix $S$ is given by (\ref{ES}). In combination with (\ref{cXdot}), the relation (\ref{SAAS}) leads to the ODE
\begin{equation}
\label{Edot}
  \dot{E}
  =
  S \dot{\cX}
  =
  S \cA \cX
  =
  \wh{\cA}S \cX
  =
  \wh{\cA} E.
\end{equation}
These autonomous dynamics preserve the CCRs for the estimation error:
\begin{equation}
\label{EEcomm}
    [E,E^{\rT}]
    =
    2i \wh{\Theta},
    \qquad
    \wh{\Theta}:= S \Theta S^\rT
    =
    \sum_{k=1}^2
    S_k \Theta_k S_k^\rT.
\end{equation}
Indeed, from (\ref{SAAS}), (\ref{EEcomm}) and the Hamiltonian property $\cA \in \Theta \mS_{n+\nu}$,  it follows that
\begin{equation}
\label{PRE}
    \wh{\cA} \wh{\Theta} + \wh{\Theta} \wh{\cA}^\rT
     =
    \wh{\cA} S \Theta S^\rT + S \Theta S^\rT \wh{\cA}^\rT     \\
    =
    S (\cA \Theta + \Theta \cA^\rT) S^\rT = 0.
\end{equation}
Therefore, if the CCR matrix $\wh{\Theta} \in \mA_p$ in (\ref{EEcomm}) is nonsingular, then (\ref{PRE}) implies that $\wh{\cA}$ is Hamiltonian in the sense that $ \wh{\cA} \in \wh{\Theta} \mS_p$.

Now, let the plant and the observer have equal dimensions $n=\nu$ and identical CCR matrices
\begin{equation}
\label{TTT}
    \Theta_0 := \Theta_1 = \Theta_2,
\end{equation}
with $\Theta_0 \in \mA_n$ and $\det \Theta_0\ne 0$.
Also, suppose the estimation error $E$ in (\ref{ES}) has the same dimension $p = n$ and is specified by equal nonsingular matrices
\begin{equation}
\label{SSS}
    S_0 := S_1 = S_2,
\end{equation}
with $S_0 \in \mR^{n\x n}$ and $\det S_0\ne 0$. Then the process $E$  reduces to
\begin{equation}
\label{ES0}
    E = S_0(X-\xi),
\end{equation}
and its CCR matrix in (\ref{EEcomm}) is nonsingular:
\begin{equation}
\label{Thetahat}
    \wh{\Theta}
    =
    2 S_0 \Theta_0 S_0^\rT.
\end{equation}
Since $E^{\rT} E = (X-\xi)^{\rT}S_0^\rT S_0 (X-\xi)$ in view of (\ref{ES0}), the matrix $ S_0^\rT S_0 \succ 0$ specifies the relative importance of the plant variables in the CQF problem (\ref{cZ}), (\ref{ZE}).
\begin{lem}
\label{lem:KML}
Under the conditions (\ref{TTT}) and (\ref{SSS}), the estimation error (\ref{ES0}) acquires the autonomous dynamics (\ref{Edot}) due to (\ref{SAAS}) for some matrix $\wh{\cA} \in \mR^{n\x n}$ if and only if the observer has the same energy matrix as the plant and a symmetric coupling matrix:
\begin{equation}
\label{KML}
    K = M,
    \qquad
    L = L^\rT.
\end{equation}
For any such observer, the matrix $\wh{\cA}$ is found uniquely as
\begin{equation}
\label{Ahat}
    \wh{\cA}
    =
    2\wh{\Theta} \wh{R},
\end{equation}
where $\wh{\Theta}$ is the CCR matrix of the estimation error in (\ref{Thetahat}), and
\begin{equation}
\label{Rhat}
    \wh{R}
    :=
    \tfrac{1}{2}
    S_0^{-\rT}
    (K-L)
    S_0^{-1}
\end{equation}
is a real symmetric matrix of order $n$.
\end{lem}
\begin{proof}
A combination of (\ref{cA}) with (\ref{TTT}) and (\ref{SSS}) leads to
\begin{equation}
\label{SAAS1}
    S\cA
     =
    2S_0 \Theta_0
    {\scriptsize\begin{bmatrix}
        K - L^\rT &  L - M
    \end{bmatrix}},
    \quad
    \wh{\cA} S
    =
    \wh{\cA} S_0
    {\scriptsize\begin{bmatrix}
         I_n &  -I_n
    \end{bmatrix}}.
\end{equation}
Therefore, since $\det S_0\ne 0$, the fulfillment of (\ref{SAAS}) for some matrix $\wh{\cA} \in \mR^{n\x n}$  is equivalent to
$
    K - L^\rT = M-L
$, that is,
\begin{equation}
\label{LKM}
  M - K = L-L^{\rT}.
\end{equation}
Since the left-hand side of (\ref{LKM}) is a symmetric matrix, while its right-hand side is antisymmetric, and only the zero matrix has these properties simultaneously ($\mS_n \bigcap \mA_n = \{0\}$), then (\ref{LKM}) holds if and only if $L$ and $M$ satisfy (\ref{KML}). In this case,  (\ref{SAAS}), (\ref{Thetahat}) and (\ref{SAAS1}) imply that
$
    \wh{\cA} = 2S_0 \Theta_0 (K-L) S_0^{-1} = \wh{\Theta} S_0^{-\rT}(K-L) S_0^{-1}
$,
which leads to (\ref{Ahat}), with $\wh{R}$ given by (\ref{Rhat}).
\end{proof}

The observer, described in Lemma~\ref{lem:KML}, replicates the quantum plant, except that it is endowed with a different initial space and, in general, different initial covariance conditions in (\ref{Sig}).  The structure (\ref{KML}) of such observers does not depend on particular matrices  $\Theta_0$ and $S_0$.
In view of (\ref{Edot}) and (\ref{Ahat}), 
the entries of the estimation error $E$ in (\ref{ES0})  evolve in time as system variables of a QHO with the CCR matrix $\wh{\Theta}$ in (\ref{Thetahat}) and the energy matrix $\wh{R}$ in (\ref{Rhat}). Without additional constraints on the coupling matrix $L$ (apart from its symmetry in (\ref{KML})), $\wh{R}$ can be ascribed any given value in $\mS_n$ by an appropriate choice of $L$. However, large values of $L$ are penalized by the second term of the cost functional in (\ref{ZE}). A solution of the CQF problem (\ref{cZ}) in this class of observers is as follows.

\begin{thm}
\label{th:Lopt}
In the framework of Lemma~\ref{lem:KML} under the conditions (\ref{TTT}), (\ref{SSS}) and $\cP_{22}\succ 0$, an optimal coupling matrix $L \in \mS_n$ for the observer with autonomous estimation error dynamics satisfies
\begin{equation}
\label{LLL}
  L
  =
  -\tfrac{8}{\lambda}
  \Pi^{-1}
  \bL(\cP_{22}\Pi^{-1}, \bS(\bS(\Theta \cE)_{12}))
  \Pi^{-1}.
\end{equation}
\end{thm}
\begin{proof}
In view of (\ref{KML}), the observer energy matrix $M=K$ remains fixed, and, due to the symmetry of  $L$,  the first variation (\ref{delcZ}) of the  cost functional in the proof of Theorem~\ref{th:stat} reduces to
$
    \delta \cZ = 2 \Bra \bS(\lambda \Pi L\cP_{22} - 4\bS(\Theta \cE)_{12}), \delta L\Ket
$. Hence,
\begin{align}
\nonumber
    \d_L \cZ
    & =
    2 \bS(\lambda \Pi L\cP_{22} - 4\bS(\Theta \cE)_{12})\\
\nonumber
    & =
    \lambda
    (\cP_{22} L \Pi
    +
    \Pi L\cP_{22}
    )
    -
    8 \bS(\bS(\Theta \cE)_{12})\\
\label{dZdLsym}
    & =
    \lambda
    (\cP_{22}\Pi^{-1}\wt{L}
    +
    \wt{L}\Pi^{-1}\cP_{22}
    )
    -
    8 \bS(\bS(\Theta \cE)_{12}),
\end{align}
where
\begin{equation}
\label{LPi}
    \wt{L}:= \Pi L \Pi
\end{equation}
inherits its symmetry from $L$ and $\Pi$. From (\ref{dZdLsym}),  it follows  that $\d_L \cZ =0$ is equivalent to $\wt{L}$ being a unique solution of an appropriate ALE:
\begin{equation}
\label{LALE}
    \wt{L}
    =
  -\tfrac{8}{\lambda}
  \bL(\cP_{22}\Pi^{-1}, \bS(\bS(\Theta \cE)_{12})),
\end{equation}
where $\cP_{22}\Pi^{-1}$ is isospectral to $\Pi^{-1/2}\cP_{22}\Pi^{-1/2}\succ 0$. A combination of (\ref{LPi}) with (\ref{LALE}) leads to (\ref{LLL}).
\end{proof}

The right-hand side of the equation (\ref{LLL}) is a nonlinear composite function of the coupling matrix $L$ and a scalar parameter
\begin{equation}
\label{mu}
    \mu:= \tfrac{1}{\lambda}>0
\end{equation}
(which is assumed to be sufficiently small), and can be represented as
\begin{equation}
\label{LF}
    L_{\mu} = \mu f(\mu, L_{\mu}).
\end{equation}
The computation of the function $f$ involves the solution of the ALEs (\ref{cPALE}) and (\ref{cQALE}) for the Gramians   $\cP$ and $\cQ$ with the matrix
\begin{equation}
\label{cAL}
    \cA
    =
    2
    {\scriptsize\begin{bmatrix}
        \Theta_0 K & \Theta_0 L\\
        \Theta_0 L & \Theta_0 K
    \end{bmatrix}},
\end{equation}
followed by computing the Hankelian $\cE$ in (\ref{cE}) and solving the ALE (\ref{LALE}). 
The parameter $\mu$ in (\ref{mu}) enters $f$ only through the matrix
\begin{equation}
\label{CCL}
    \cC^\rT \cC
    =
    S^\rT S +
    {\scriptsize\begin{bmatrix}
        0 & 0\\
        0 & \frac{1}{\mu}L \Pi L
    \end{bmatrix}}
\end{equation}
in the ALE (\ref{cQALE}).  The smallness of $\mu$ corresponds to large values of $\lambda$ (that is, high penalization of the observer back-action on the plant). For all sufficiently small $\mu>0$ and $L \in \mS_n$, the function $f$ is Frechet differentiable, and this  smoothness is inherited by $L_\mu$ in (\ref{LF}). The differentiation of (\ref{LF}) with respect to $\mu$ (as fictitious time) leads to the ODE 
\begin{equation}
\label{Ldot}
    \d_\mu L_\mu = (\cI -\mu \d_L f )^{-1}(f + \mu \d_\mu f),
\end{equation}
with the initial condition $L_0 = 0$,
where $\cI$ is the identity operator on the space $\mS_n$, and $\d_L f$ is the appropriate partial Frechet derivative of $f$. The initial-value problem (\ref{Ldot}) describes a homotopy method for numerical solution of the CQF problem, similar to \cite{MB_1985} (see also, \cite{VP_2010b}). The right-hand side of (\ref{Ldot}) is well-defined for all $(\mu,L)$ in a small neighbourhood of $(0,0)$.  Its computation can be implemented by using the vectorised representations of the Frechet derivatives of solutions of ALEs \cite{SIG_1998,VP_2013a}  in application to the ALEs (\ref{cPALE}), (\ref{cQALE}) (or their Lie-algebraic forms (\ref{Pres}), (\ref{Qres})) and (\ref{LALE}). The details of these calculations are tedious and omitted for brevity. The weak-coupling (or high-penalization) asymptotic behaviour of the matrix $L_{\mu}$ is described below.

\begin{thm}
\label{th:asy}
Suppose the uncoupled observer has a positive definite matrix $P_2$ in (\ref{PALE2}).  Then, for large values of the parameter $\lambda$ in (\ref{ZE}), the optimal coupling  matrix in Theorem~\ref{th:Lopt} satisfies the asymptotic relation
\begin{equation}
\label{Lasy}
  L_\mu \sim \mu L',
  \qquad
  {\rm as}\
  \mu \to 0+,
\end{equation}
where the matrix $L' \in \mS_n$ is a unique solution of the ALE
\begin{equation}
\label{L'}
  L'
  =
  2
  \Pi^{-1}
  \bL
  \big(
  P_2 \Pi^{-1},
  \Theta_0 \cQ_0 (P_1+P_2)
  -
  (P_1+P_2) \cQ_0  \Theta_0
  \big)
  \Pi^{-1}.
\end{equation}
Here, $P_1$ and $P_2$ are the second-moment matrices (\ref{PALE1}) and (\ref{PALE2}) for the uncoupled plant and observer variables given by
\begin{align}
\label{PALE0}
    P_k &= \tfrac{1}{\tau}\bL(A_{\tau}, \Sigma_k),
    \qquad
    k = 1,2,
\end{align}
with a common matrix
\begin{equation}
\label{A0}
    A_{\tau}:= A - \tfrac{1}{2\tau}I_n,
    \qquad
    A= 2\Theta_0 K.
\end{equation}
Also,
\begin{equation}
\label{cQ0}
  \cQ_0
  :=
  S_0^\rT \wh{\cQ} S_0
\end{equation}
in (\ref{L'}) is associated with a unique solution $\wh{\cQ}$ of the ALE
\begin{equation}
\label{Qhat}
  \wh{\cQ}
  :=
  \bL(\wh{\cA}_{\tau}^\rT, I_n),
  \quad
  \wh{\cA}_{\tau}:= \wh{\cA} - \tfrac{1}{2\tau}I_n,
  \quad
  \wh{\cA}= 2S_0 \Theta_0 K S_0^{-1}.
\end{equation}
\end{thm}
\begin{proof}
From the representation (\ref{LF}) of (\ref{LLL}) (or from (\ref{Ldot})), it follows that (\ref{Lasy}) holds with
\begin{equation}
\label{L''}
    L':= \d_\mu L_{\mu}\big|_{\mu = 0} = f(0,0),
\end{equation}
where we have also used the initial condition $L_0=0$. Here,
\begin{equation}
\label{f00}
    f(0,0)
    =
  -8
  \Pi^{-1}
  \bL(P_2 \Pi^{-1}, \bS(\bS(\Theta \cE)_{12}))
  \Pi^{-1}
\end{equation}
is associated with the uncoupled plant and observer, in which case they have the block-diagonal controllability Gramian in (\ref{cP*}), where the matrices $P_1$ and $P_2$ are given by (\ref{PALE0}), (\ref{A0}) since the matrix (\ref{cAL}) reduces to $    \cA
    =
    2I_2 \ox (\Theta_0 K)
$ with a purely imaginary spectrum due to $K\succ 0$.
In the limit of uncoupled plant and observer,
$\tfrac{1}{\mu }L_{\mu} \Pi L_\mu \sim \mu f(0,0)\Pi f(0,0)\to 0 $ as $\mu\to 0+$, whereby
(\ref{CCL}) leads to $\cC^\rT\cC=S^\rT S$ at $\mu=0$, and the ALE (\ref{cQALE}) for the observability Gramian $\cQ$  takes the form
\begin{equation}
\label{cQ0ALE}
    \cA_\tau^\rT \cQ
    +
    \cQ \cA_\tau
    +
    S^\rT S = 0.
\end{equation}
The property (\ref{SAAS}) of the observers under consideration implies that $S\cA_{\tau} = S\cA  - \tfrac{1}{2\tau} S = \wh{\cA} S  - \tfrac{1}{2\tau} S = \wh{\cA}_{\tau}S$ and hence,
 (\ref{cQ0ALE}) admits a lower-rank solution
\begin{equation}
\label{SQS}
    \cQ = S^\rT \wh{\cQ} S.
\end{equation}
Indeed, its substitution into the left-hand side of (\ref{cQ0ALE}) yields
$
    \cA_\tau^\rT \cQ
    +
    \cQ \cA_\tau
    +
    S^\rT S
     =
    S^\rT(\wh{\cA}_\tau^\rT \wh{\cQ}
    +
    \wh{\cQ} \wh{\cA}_\tau
    +
    I_n)
    S$. Therefore, (\ref{Qhat}) makes (\ref{SQS}) a unique solution of the ALE (\ref{cQ0ALE}), since the matrix $\wh{\cA}$ is isospectral to $2\Theta_0 K$ with a purely imaginary spectrum (so that $\wh{\cA}_\tau$ is Hurwitz). In view of $S = {\scriptsize\begin{bmatrix} 1 & -1\end{bmatrix}}\ox S_0$, the Hankelian takes the form
    $
        \cE
        =
        S^\rT \wh{\cQ} S \diag_{k=1,2}(P_k)
        =
        {\scriptsize
            \begin{bmatrix}
            \cQ_0 P_1 & - \cQ_0 P_2\\
            -\cQ_0P_1 & \cQ_0P_2
            \end{bmatrix}}
    $,
    with the matrix $\cQ_0$ given by (\ref{cQ0}),
    and hence,
    \begin{equation}
    \label{SST}
    \bS(\bS(\Theta \cE)_{12})
    =
    \tfrac{1}{4}
    ((P_1+P_2)\cQ_0 \Theta_0 - \Theta_0 \cQ_0 (P_1+P_2)).
    \end{equation}
    Substitution of  (\ref{SST}) into  (\ref{f00}) and (\ref{L''}) leads to (\ref{L'}).
\end{proof}

\noindent{\bf Example 2.}
Let the plant and observer be one-mode QHOs ($n=\nu=2$), with the CCR matrix  $\Theta_0 := \frac{1}{2} \bJ$ (corresponding to the position-momentum pair, with $\bJ$ given by (\ref{bJ}), and a positive definite energy matrix
    $$
    K :=
        {\scriptsize\begin{bmatrix}
         2.7604 &  -1.7564\\
        -1.7564 &   2.4982
        \end{bmatrix}}.
    $$
The frequencies of such a QHO are $\pm 1.9522$, and the corresponding margin (\ref{tau*}) is   $\tau_* = 0.2561$.
The initial covariance conditions in (\ref{Sig})  for the plant and observer (prepared independently) are
   $$
    \Sigma_1
    ={\scriptsize\begin{bmatrix}
    4.1400 &  -2.4687\\
   -2.4687 &   4.3641
    \end{bmatrix}},
    \qquad
    \Sigma_2
    ={\scriptsize\begin{bmatrix}
    2.2174  &  1.3387\\
    1.3387  &  2.4695
    \end{bmatrix}}
    $$
    and satisfy the uncertainty relation constraints $\Sigma_k+i\Theta_0 \succcurlyeq 0$.
    With the effective time horizon chosen to be    $\tau=4.0614\gg \tau_*$,
the second-moment matrices of the uncoupled plant and observer variables in (\ref{PALE1}), (\ref{PALE2}) are
    $$
    P_1 =
    {\scriptsize\begin{bmatrix}
    9.7049 &     7.0975\\
    7.0975 &   11.6664
    \end{bmatrix}},
    \qquad
    P_2 =
    {\scriptsize\begin{bmatrix}
    2.4681 &    1.7476\\
    1.7476 &   2.7674
    \end{bmatrix}}.
    $$
    For the     CQF problem (\ref{cZ}),
    the observer back-action penalty matrix $\Pi$ in (\ref{ZE}) and the weighting matrix $S_0$ in the estimation error (\ref{ES0}) are given by
    $$
    \Pi =
    {\scriptsize\begin{bmatrix}
     1.2907 &   0.9694\\
    0.9694  &  3.7716
    \end{bmatrix}},
    \qquad
    S_0 =
    {\scriptsize\begin{bmatrix}
       -1.7389  &  0.2192\\
    0.0170  &  1.0458
    \end{bmatrix}}.
    $$
    The mean square of the estimation error for the uncoupled observer is
    $
        \Tr(S_0(P_1 + P_2)S_0^\rT) = 46.8634$.
    The mean square value $\bE_{\tau}(E^\rT E)$ for the optimal observer in the CQF problem (\ref{cZ}) (subject to the autonomous estimation error dynamics) is shown in Fig.~\ref{fig:EEE}
     \begin{figure}[thpb]
      \centering
      \includegraphics[width=85mm,height=55mm]{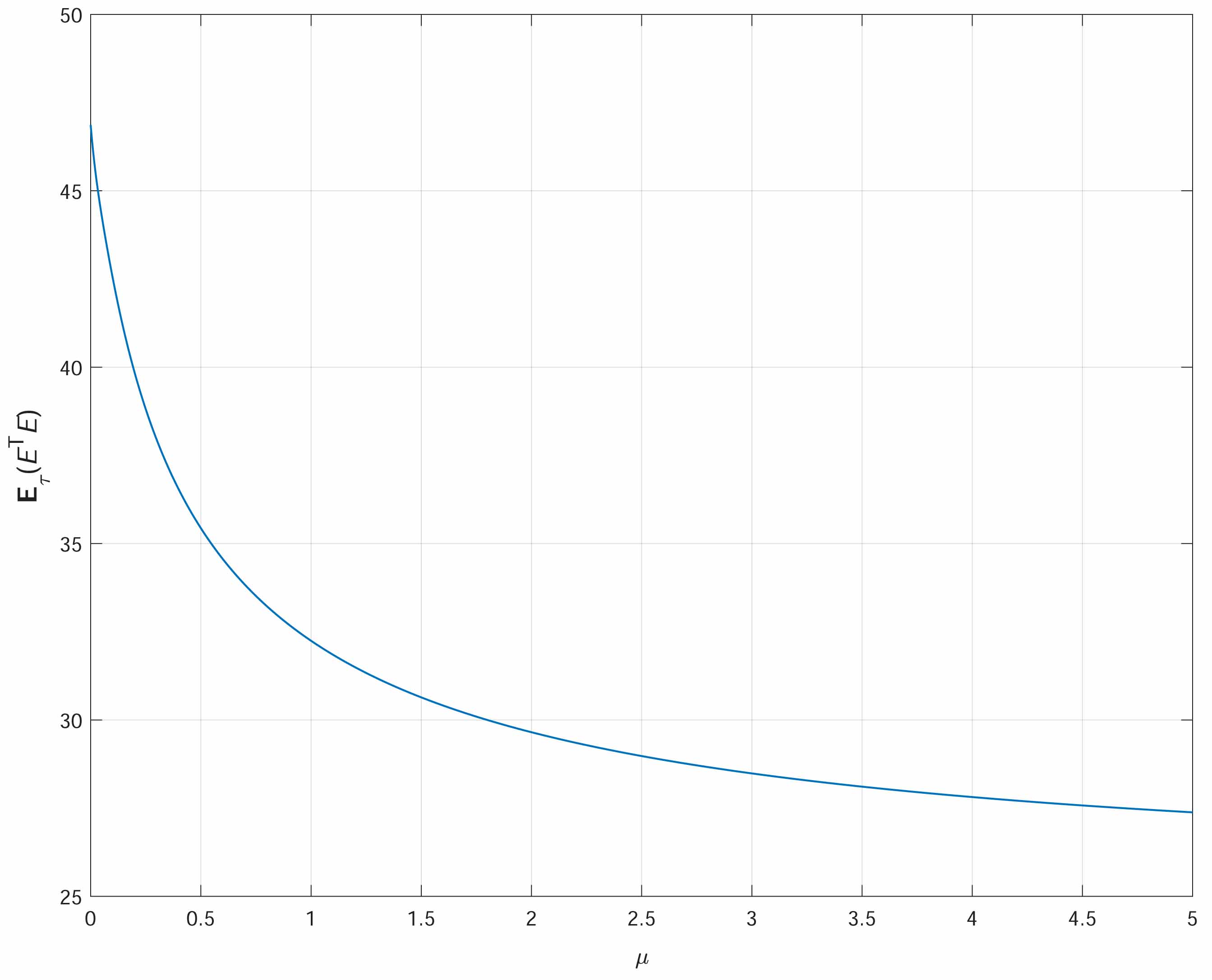}
      \caption{The dependence of the mean square value $\bE_\tau(E^\rT E)$ of the estimation error  on the parameter $\mu$ in Example 2.}
      \label{fig:EEE}
   \end{figure}
for a range of values of the parameter $\mu$ in (\ref{mu}). This is a monotonically decreasing function of $\mu$, whose computation (along with the optimal observers) was carried  out  using the homotopy method starting from the uncoupled observer (at $\mu = 0$). The observer back-action penalty term $    \lambda \bE_\tau
    (\eta^{\rT} \Pi \eta)$ is shown in Fig.~\ref{fig:pen}.
\begin{figure}[thpb]
      \centering
      \includegraphics[width=85mm,height=55mm]{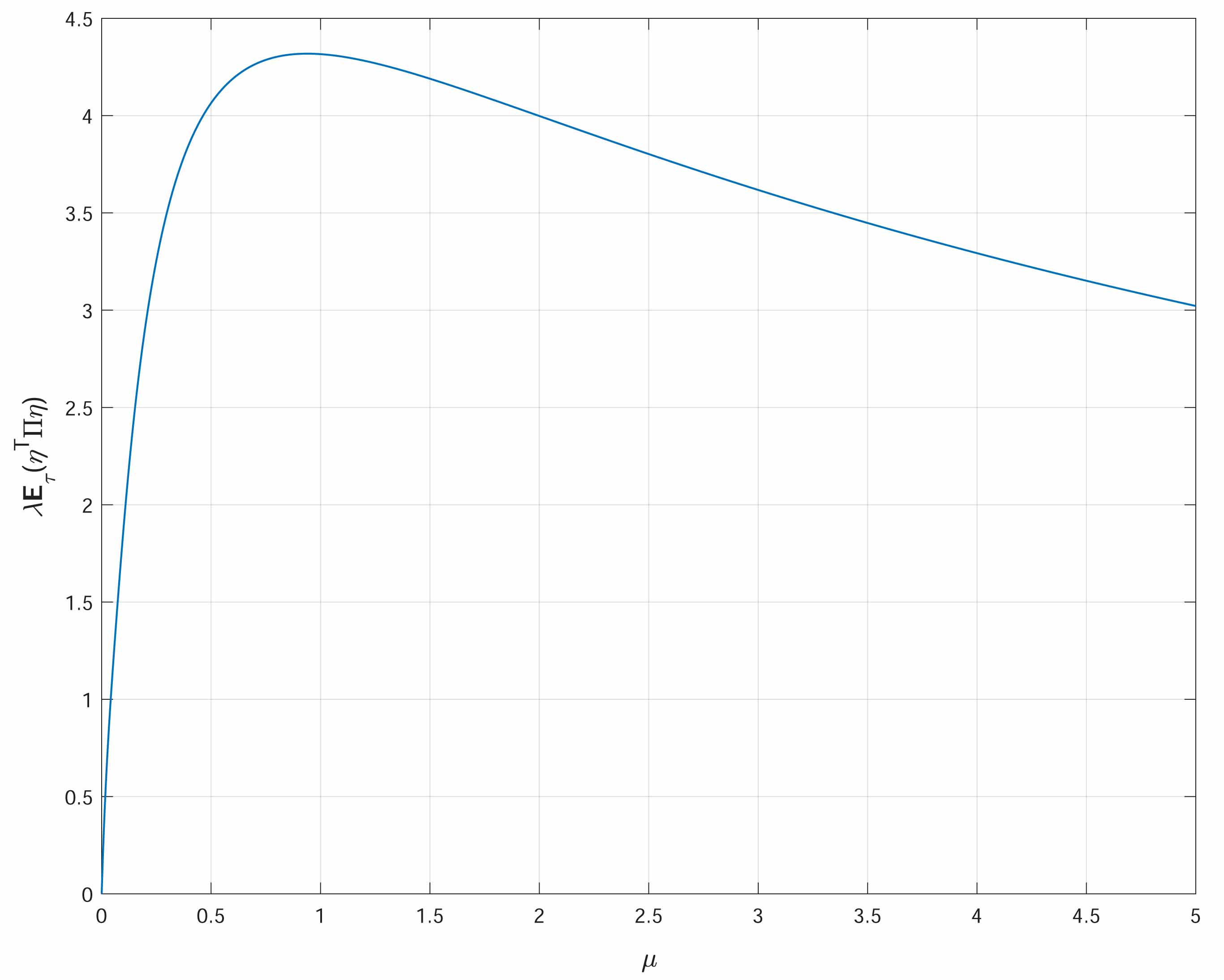}
      \caption{The dependence of the back-action penalty term $    \lambda \bE_\tau
    (\eta^{\rT} \Pi \eta)$ on the parameter $\mu$ in Example 2.}
      \label{fig:pen}
   \end{figure}
    The entries of the corresponding optimal coupling matrix $L_\mu$  are presented in Fig.~\ref{fig:LLL}.
\begin{figure}[thpb]
      \centering
\includegraphics[width=85mm,height=55mm]{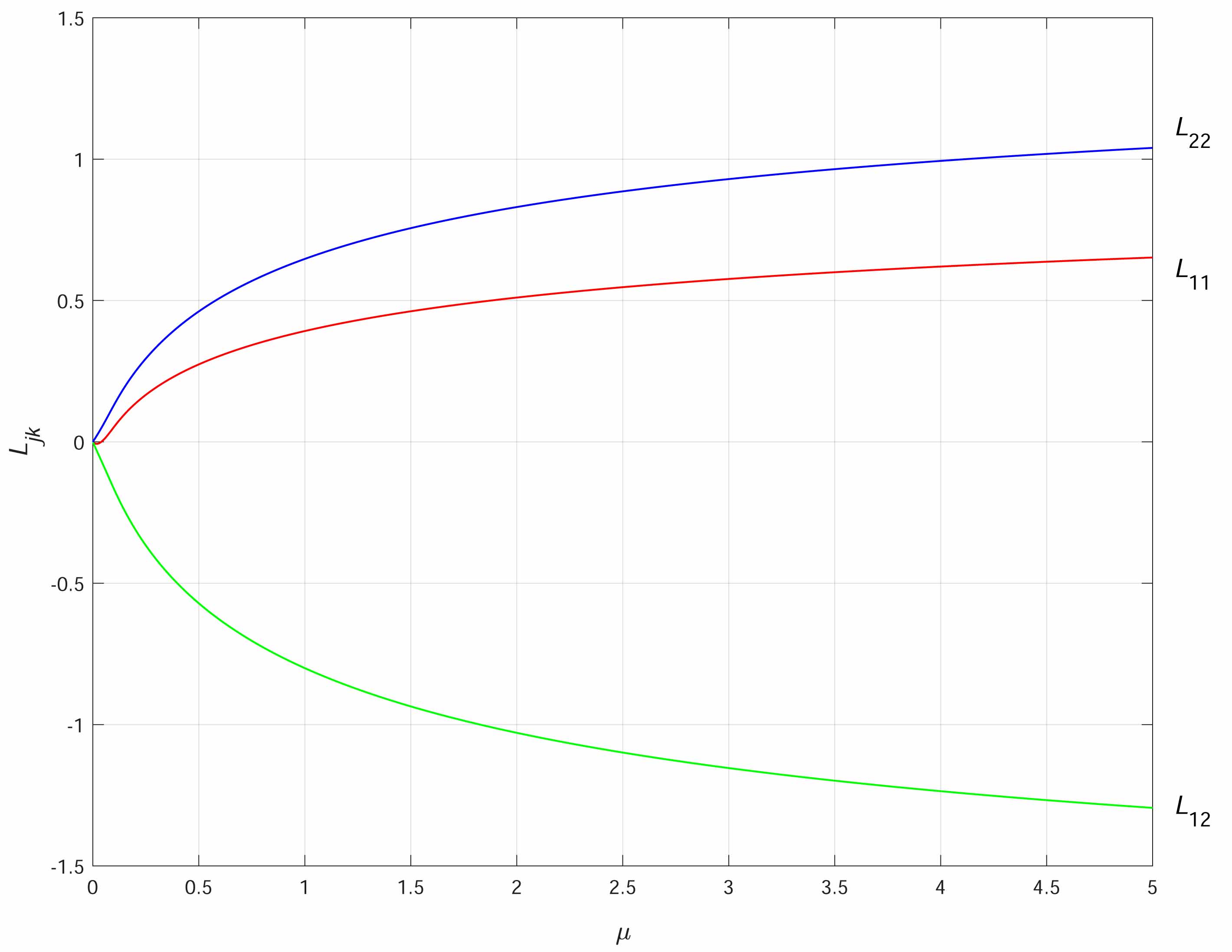}      
      
      \caption{The dependence of the entries of the coupling matrix $L_{\mu}$ on the parameter $\mu$ in Example 2.}
      \label{fig:LLL}
   \end{figure}
   The calculation of the matrix $L'$, which specifies their asymptotic behaviour as $\mu\to 0+$
    according to Theorem~\ref{th:asy}, yielded
    $$
    L'
    =
    {\scriptsize\begin{bmatrix}
    -0.7297 &  -1.7445\\
   -1.7445  &  1.1737
   \end{bmatrix}}.
    $$
    For the range $0\< \mu \< 5$  (that is, $\lambda \>0.2$), the plant-observer energy matrix ${\scriptsize\begin{bmatrix}K & L_\mu\\ L_\mu & K\end{bmatrix}}$ remained positive definite, so that the system variables retained oscillatory behaviour, which justifies the discounted averaging approach.
\hfill$\blacktriangle$

\section{Conclusion}\label{sec:conc}

We have considered the computation of discounted averages with exponentially decaying weights for moments of system variables for QHOs, including the mean square functionals, both in the state space and frequency domain.  For a quantum plant and a quantum observer in the form of directly coupled  QHOs,  we have obtained small-gain-theorem  bounds  for the back-action of the observer on the covariance dynamics of the plant in terms of the plant-observer coupling. We have considered a CQF problem of minimizing the discounted mean square value of the estimation error together with a penalty on the observer back-action.  First-order necessary conditions of optimality have been obtained for this problem in the form of a set of algebraic matrix equations involving two coupled ALEs. We have applied Lie-algebraic 
techniques  to these equations and discussed a solution of the CQF problem in the case of autonomous estimation error  dynamics, including the homotopy method for its implementation.  These results  have been illustrated by numerical experiments. 




\begin{thebibliography}{99}
\bibitem{AM_1979}
B.D.O.Anderson, and J.B.Moore,
\emph{Optimal Filtering},
Prentice Hall, New York, 1979.
\bibitem{AM_1989}
B.D.O.Anderson, and J.B.Moore,
\emph{Optimal Control: Linear Quadratic Methods},
Prentice Hall, London, 1989.
\bibitem{A_1989}
V.I.Arnold, \emph{Mathematical Methods of Classical Mechanics}, 2nd Ed., Springer-Verlag, New York, 1989.
\bibitem{BH_1989}
D.S.Bernstein, and W.M.Haddad,
LQG control with an
$H^{\infty}$ performance bound: a Riccati equation approach,
\textit{IEEE Trans.
Automat. Contr.}, vol. 34, no. 3, 1989, pp. 293--305.
\bibitem{B_1954}
A.S.Besicovitch,
\emph{Amost Periodic Functions}, Dover, New York, 1954.
\bibitem{B_1968}
P.Billingsley, \emph{Convergence of Probability Measures}, John Wiley \& Sons, New York, 1968.
\bibitem{B_1965}
    D.Blackwell, Discounted dynamic programming, \emph{Ann. Math. Statist.},
    vol. 36, no. 1, 1965, pp. 226--235.
\bibitem{BP_2006}
H.-P.Breuer, and F.Petruccione,
\textit{The Theory of Open Quantum Systems},
Clarendon Press, Oxford, 2006.
\bibitem{CL_1981}
A.O.Caldeira, and A.J.Leggett, Influence of dissipation on quantum tunneling in macroscopic systems, \textit{Phys. Rev. Lett.}, vol. 46, no. 4, 1981, p. 211--214.
\bibitem{CH_1971}
C.D.Cushen, and R.L.Hudson, A quantum-mechanical central limit theorem,
\emph{J. Appl. Prob.}, vol. 8, no. 3, 1971, pp. 454--469.
\bibitem{DJ_2006}
C.D'Helon,  and M.R.James,
Stability, gain, and robustness in quantum feedback networks,
\emph{Phys. Rev. A.}, vol. 73, no. 5, 2006, p. 053803.
\bibitem{D_2006}
M. de Gosson,
\emph{Symplectic Geometry and Quantum Mechanics},
Birkh\"{a}user, Basel, 2006.
\bibitem{EB_2005}
S.C.Edwards, and V.P.Belavkin,
Optimal quantum filtering and
quantum feedback control,
arXiv:quant-ph/0506018v2, August 1,  2005.
\bibitem{F_1971}
W.Feller, \emph{An Introduction to Probability Theory and Its Applications. Vol. II}, 2nd Ed., John Wiley \& Sons, New York, 1971.


\bibitem{GZ_2004}
C.W.Gardiner, and P.Zoller,
\textit{Quantum Noise}, 3rd Ed.,
Springer, Berlin, 2004.
\bibitem{GJ_2009}
J.Gough, and M.R.James,
Quantum feedback networks: Hamiltonian
formulation,
\emph{Commun. Math. Phys.},  vol. 287, 2009, pp. 1109--1132.
\bibitem{H_2008}
N.J.Higham,
\emph{Functions of Matrices}, SIAM, Philadelphia, 2008.

\bibitem{H_1991}
A.S.Holevo, Quantum stochastic calculus,
\textit{J. Math. Sci.}, vol. 56, no. 5, 1991, pp. 2609--2624.
\bibitem{H_2001}
A.S.Holevo, \textit{Statistical Structure of Quantum Theory}, Springer, Berlin, 2001.
\bibitem{HJ_2007}
R.A.Horn, and C.R.Johnson,
\textit{Matrix Analysis},
Cambridge
University Press, New York, 2007.
\bibitem{HP_1984}
R.L.Hudson,  and K.R.Parthasarathy,
Quantum Ito's Formula and Stochastic Evolutions,
\emph{Commun. Math. Phys.}, vol.  93, 1984, pp. 301--323.
\bibitem{I_1918}
L.Isserlis,
On a formula for the product-moment coefficient of any order of a normal frequency
distribution in any number of variables,
     \textit{Biometrika}, vol. 12,  1918, pp. 134--139.


\bibitem{JNP_2008}
M.R.James, H.I.Nurdin, and I.R.Petersen,
$H^{\infty}$ control of
linear quantum stochastic systems,
\textit{IEEE Trans.
Automat. Contr.}, vol. 53, no. 8, 2008, pp. 1787--1803.
\bibitem{J_1997}
S.Janson,
\textit{Gaussian Hilbert Spaces},
Cambridge University Press, Cambridge, 1997.

\bibitem{KS_1972}
H.Kwakernaak, and R.Sivan,
\textit{Linear Optimal Control Systems},
Wiley, New York, 1972.
\bibitem{L_2000}
S.Lloyd, Coherent quantum feedback, \emph{Phys. Rev. A}, vol. 62, no. 2, 2000, pp. 022108.
\bibitem{M_1988}
J.R.Magnus,
\textit{Linear Structures},
Oxford University Press, New York, 1988.
\bibitem{MB_1985}
M.Mariton,  and P.Bertrand, A homotopy algorithm for solving coupled Riccati equations,
\emph{Optim. Contr. Appl. Meth.}, vol. 6, no. 4, 1985, pp. 351--357.

\bibitem{M_1998}
E.Merzbacher,
\textit{Quantum Mechanics}, 3rd Ed.,
Wiley, New York, 1998.
\bibitem{M_1995}
P.-A.Meyer, \emph{Quantum Probability for Probabilists}, 2nd Ed., Springer, Berlin, 1995.
\bibitem{MJ_2012}
Z.Miao, and M.R.James,
Quantum observer for linear quantum
stochastic systems, Proc. 51st IEEE Conf. Decision Control, Maui,
Hawaii, USA, December 10-13, 2012, pp. 1680--1684.

\bibitem{NC_2000}
M.A.Nielsen, and I.L.Chuang,
\textit{Quantum Computation and Quantum Information},
Cambridge University Press, Cambridge, 2000.

\bibitem{NJP_2009}
H.I.Nurdin, M.R.James, and I.R.Petersen,
Coherent quantum LQG
control,
\textit{Automatica}, vol.  45, 2009, pp. 1837--1846.

\bibitem{O_1993}
P.J.Olver,
\textit{Applications of Lie Groups to Differential Equations}, 2nd Ed.,
Springer, New York, 1993.

\bibitem{P_1992}
K.R.Parthasarathy,
\textit{An Introduction to Quantum Stochastic Calculus},
Birk\-h\"{a}user, Basel, 1992.
\bibitem{KRP_2010}
K.R.Parthasarathy,
What is a Gaussian state?
\emph{Commun. Stoch. Anal.}, vol. 4, no. 2, 2010, pp. 143--160.
\bibitem{P_2010}
I.R.Petersen,
Quantum linear systems theory,
Proc. 19th Int. Symp. Math. Theor. Networks Syst., Budapest, Hungary, July 5--9, 2010, pp.  2173--2184.
\bibitem{P_2014}
I.R.Petersen,
A direct coupling coherent quantum observer, Proc.
IEEE MSC 2014, Nice/Antibes, France, 8--10 October 2014, pp. 1960--1963.
\bibitem{P_2017}
I.R.Petersen,
Time averaged consensus in a direct coupled coherent quantum observer network,
\textit{Contr. Theory Techn.}, vol. 15, no. 3, pp. 163--176.
\bibitem{PH_2015}
I.R.Petersen, and E.H.Huntington,
A possible implementation of a direct coupling coherent quantum
observer,
Proc. Australian Control Conference, 5-6 November 2015, Gold Coast, Australia, pp. 105--107
(arXiv:1509.01898v2 [quant-ph], 10 September 2015).
\bibitem{PBGM_1962}
L.S.Pontryagin, V.G.Boltyanskii,
R.V.Gamkrelidze, and E.F. Mishchenko,
\emph{The Mathematical Theory of Optimal Processes},
Wiley, New York, 1962.

\bibitem{S_1994}
J.J.Sakurai,
\textit{Modern Quantum Mechanics},
 Addison-Wesley, Reading, Mass., 1994.

\bibitem{SP_2009}
A.J.Shaiju, and I.R.Petersen,
On the physical realizability of
general linear quantum stochastic differential equations with
complex coefficients,
Proc. Joint 48th IEEE Conf. Decision Control \&
28th Chinese Control Conf.,
Shanghai, P.R. China, December 16--18, 2009, pp. 1422--1427.
\bibitem{SP_2012}
A.J.Shaiju, and I.R.Petersen,
A frequency domain condition for the physical realizability of linear quantum systems,
\emph{IEEE Trans. Automat. Contr.},  vol. 57, no. 8, 2012, pp. 2033--2044.
\bibitem{SVP_2014}
A.K.Sichani, I.G.Vladimirov, and I.R.Petersen,
Robust mean square stability of open quantum stochastic systems with Hamiltonian perturbations in a Weyl quantization form,
Proc. Australian Control Conference
Canberra, 17-18 November 2014, pp. 83--88 (arXiv:1503.02122 [quant-ph], 7 March 2015).

\bibitem{SVP_2017}
A.K.Sichani, I.G.Vladimirov, and I.R.Petersen,
A numerical approach to optimal coherent quantum LQG controller design using gradient descent, \textit{Automatica}, vol. 85, 2017, pp. 314--326.


\bibitem{SIG_1998}
R.E.Skelton, T.Iwasaki, and K.M.Grigoriadis,
\textit{A Unified Algebraic Approach to Linear Control Design},
Taylor \& Francis, London, 1998.
\bibitem{SW_1997}
H.J.Sussmann, and J.C.Willems,
300 years of optimal control: from the brachystochrone to the maximum principle,
\textit{Control Systems}, vol. 17, no. 3, 1997, pp. 32--44.
\bibitem{V_1984}
V.S.Varadarajan, \emph{Lie Groups, Lie Algebras, and Their Representations},
Springer-Verlag, New York, 1984.

\bibitem{VP_2010b}
I.G.Vladimirov, and I.R.Petersen,
Hardy-Schatten norms of systems, output energy cumulants and linear quadro-quartic  Gaussian control,
Proc. 19th Int. Symp. Math. Theor. Networks Syst., Budapest, Hungary, July 5--9,  2010, pp.  2383--2390.
\bibitem{VP_2011b}
I.G.Vladimirov, and I.R.Petersen,
A dynamic programming approach to finite-horizon coherent quantum LQG control, 	
Proc. Australian Control Conference, Melbourne, 10--11 November, 2011, pp. 357--362
(preprint:  	arXiv:1105.1574v1 [quant-ph], 9 May 2011).
\bibitem{VP_2013a}
I.G.Vladimirov, and I.R.Petersen,
A quasi-separation principle and Newton-like scheme for coherent quantum LQG control,
\emph{Syst. Contr. Lett.}, vol. 62, no. 7, 2013, pp. 550--559.
\bibitem{VP_2013b}
I.G.Vladimirov, and I.R.Petersen,
Coherent quantum filtering for physically realizable linear quantum plants,
Proc. European Control Conference, Zurich, Switzerland, 17-19 July 2013,  pp. 2717--2723.
\bibitem{V_2015a}
I.G.Vladimirov,
A transverse Hamiltonian  variational technique for open quantum stochastic systems and its application to coherent quantum control, Proc. IEEE Multi-Conference on Systems and Control, 21-23 September 2015, Sydney, Australia, pp. 29--34
(arXiv:1506.04737v2 [quant-ph], 7 August 2015).
\bibitem{V_2015b}
I.G.Vladimirov,
Weyl variations and local sufficiency of linear observers in the mean square optimal coherent quantum filtering problem,
Proc. Australian Control Conference, 5-6 November 2015, Gold Coast, Australia, pp. 93--98 (arXiv:1506.07653 [quant-ph], 25 June 2015).
\bibitem{VP_2016}
I.G.Vladimirov, and I.R.Petersen,
Directly coupled observers for quantum harmonic oscillators with discounted
mean square cost functionals and penalized back-action,
2016 IEEE Conference on Norbert Wiener in the 21st Century (21CW)
July 13-16, 2016. University of Melbourne, Australia, pp. 78--83 (arXiv:1602.06498 [cs.SY], 21 February 2016).
\bibitem{V_2017}
I.G.Vladimirov,
A phase-space formulation and Gaussian approximation of the filtering equations for nonlinear quantum stochastic systems,
\textit{Contr. Theory Techn.},
vol. 15, no. 3, 2017, pp. 177--192.
\bibitem{WM_1994_book}
D.F.Walls, and G.J.Milburn,
\emph{Quantum Optics},
Springer, Berlin, 1994.
\bibitem{W_1936}
J.Williamson,
On the algebraic problem concerning the normal forms of linear dynamical systems,
\textit{Am. J. Math.}, vol. 58, no. 1, 1936, pp. 141--163.
\bibitem{W_1937}
J.Williamson,
On the normal forms of linear canonical transformations in dynamics,
\textit{Am. J. Math.}, vol. 59, no. 3, 1937, pp. 599--617.
\bibitem{WM_1994_paper}
H.M.Wiseman, and G.J.Milburn, All-optical versus electro-optical quantum limited
feedback, \emph{Phys. Rev. A}, vol. 49, no. 5, 1994, pp. 4110--4125.
\end{thebibliography}
\end{document}